\documentclass[lettersize,journal]{IEEEtran}
\usepackage{amsmath,amsfonts}
\usepackage{latexsym,amsmath,amsopn,amsfonts,amssymb,amsthm}
\usepackage{algorithm}
\usepackage{algorithmic}
\usepackage{array}
\def\btimes{\,\rotatebox[]{-90}{$\ltimes$}\,}
\usepackage[caption=false,font=normalsize,labelfont=sf,textfont=sf]{subfig}
\usepackage{textcomp}
\usepackage{stfloats}
\usepackage{bm}
\usepackage{url}
\usepackage{verbatim}
\usepackage{graphicx}
\usepackage{newtxtext}
\usepackage{newtxmath}
\newtheorem{theorem}{Theorem}[section]
\newtheorem{corollary}[theorem]{Corollary}
\newtheorem{lemma}[theorem]{Lemma}
\newtheorem{definition}[theorem]{Definition}

\theoremstyle{definition}
\newtheorem{remark}[theorem]{Remark}

\def\BibTeX{{\rm B\kern-.05em{\sc i\kern-.025em b}\kern-.08em
    T\kern-.1667em\lower.7ex\hbox{E}\kern-.125emX}}
\usepackage{balance}
\begin{document}
\title{A Dimension-Keeping Semi-Tensor Product Framework for Compressed Sensing}
\author{Qi Qi, Abdelhamid Tayebi, \IEEEmembership{Fellow, IEEE}, Daizhan Cheng, \IEEEmembership{Fellow, IEEE}, and Jun-e Feng
\thanks{This work was supported by the National Natural Science Foundation of China (62273201, 62350037), and the Taishan Scholar Project of Shandong Province of China (TSTP20221103).(Corresponding author: Jun-e Feng.)

Qi Qi and Jun-e Feng are with the School of Mathematics, Shandong University, Jinan, Shandong, China 250100 (e-mail: 202411968@mail.sdu.edu.cn; fengjune@sdu.edu.cn).

Abdelhamid Tayebi is with the Department of Electrical Engineering, Lakehead University, Thunder Bay, Ontario, Canada (e-mail: atayebi@lakeheadu.ca).

Daizhan Cheng is with the Key Laboratory of Systems and Control, Academy of Mathematics and Systems Science, Chinese Academy of Sciences, Beijing, China 100190 (e-mail: dcheng@iss.ac.cn).} 
}

\markboth{IEEE TRANSACTIONS ON SIGNAL PROCESSING}%
{How to Use the IEEEtran \LaTeX \ Templates}

\maketitle

\begin{abstract}
	In compressed sensing (CS), sparse signals can be reconstructed from significantly fewer samples than required by the Nyquist-Shannon sampling theorem. While non-sparse signals can be sparsely represented in appropriate transformation domains, conventional CS frameworks rely on the incoherence of the measurement matrix columns to guarantee reconstruction performance. This paper proposes a novel method termed Dimension-Keeping Semi-Tensor Product Compressed Sensing (DK-STP-CS), which leverages intra-group correlations while maintaining inter-group incoherence to enhance the measurement matrix design. Specifically, the DK-STP algorithm is integrated into the design of the sensing matrix, enabling dimensionality reduction while preserving signal recovery capability. For image compression and reconstruction tasks, the proposed method achieves notable noise suppression and improves visual fidelity. Experimental results demonstrate that DK-STP-CS significantly outperforms traditional CS and STP-CS approaches, as evidenced by higher Peak Signal-to-Noise Ratio (PSNR) values between the reconstructed and original images. The robustness of DK-STP-CS is further validated under noisy conditions and varying sampling rates, highlighting its potential for practical applications in resource-constrained environments.
\end{abstract}


\section{Introduction}
Compressed sensing (CS), a revolutionary paradigm in signal processing, has emerged as a powerful methodology for reconstructing sparse signals from sub-Nyquist measurements \cite{WOS:000236714000001,WOS:000247642400015}. The fundamental premise of CS lies in exploiting signal sparsity, where a vector $\mathbf{x}$ is considered $k$-sparse ($\mathbf{x} \in \Sigma_k$) if it contains at most $k$ non-zero elements.

Traditional signal acquisition frameworks, governed by the Nyquist sampling theorem \cite{WOS:000174400500008}, require uniform sampling at twice the highest signal frequency. However, the exponential growth of data in modern digital systems renders Nyquist-rate sampling increasingly impractical due to excessive storage and computational demands. This limitation has motivated the development of CS, which fundamentally differs from conventional approaches through two key innovations: (1) non-uniform sampling strategies that maintain measurement incoherence, and (2) sparse signal recovery from significantly fewer measurements than required by Nyquist criteria. While CS initially assumes signal sparsity, practical implementations extend this framework to non-sparse signals through sparse representation in appropriate transformation domains \cite{WOS:A1993MU18200015}. This generalization enables efficient compression of diverse real-world signals while preserving essential information.

The standard framework of CS is
a special case of underdetermined linear equations
\begin{equation}
	\mathbf{y} = \mathbf{A x},
\end{equation}
where $ \mathbf{A} \in \mathcal{M}_{m\times n}, \mathbf{x} \in \mathbb{R}^{n},\mathbf{y} \in \mathbb{R}^{m}, m\ll n.$
Vector $\mathbf{x}$ is the input signal and vector $\mathbf{y}$ is the compressed signal. 
Vector $\mathbf{x}$ cannot be obtained directly from the above equation, but if $\mathbf{x}$ is sparse in some transformation domain, one can obtain an approximate solution for $\mathbf{x}$. In fact, vector $\mathbf{x}$ can be expressed as
\begin{equation}
	\mathbf{x} = \Theta \mathbf{s},
\end{equation}
where $\Theta$ is a sparsifying dictionary and $\mathbf{s}$ is a sparse vector. Thus,  
\begin{equation}
	\mathbf{y} = \mathbf{A}\Theta \mathbf{s} = \mathbf{\Psi s}.
\end{equation}
In the conventional compressive sensing (CS) paradigm,  measurement matrices are inherently non-adaptive, implying that the rows of the sensing matrix $\mathbf{\Psi}$ are predefined and remain invariant throughout the acquisition process, independent of any prior measurement outcomes. However, adaptive sensing strategies, where subsequent measurement matrices are dynamically adjusted based on previously observed data, have been shown in specific scenarios to achieve substantial improvements in reconstruction accuracy or sampling efficiency. 
Theoretical and practical studies highlight that such adaptability can yield enhanced performance bounds, particularly in resource-constrained or noise-dominated environments.
A solution to the underdetermined linear equation (3) can be found through the following optimization:
\begin{equation}\label{eq2}
	\min_{\mathbf{s}} \left \| \mathbf{s}  \right \|_{0}             ~~\text{subject to}~~ \mathbf{y} = \mathbf{\Psi s}, 
\end{equation}
where $\|\mathbf{s}\|_{0}$ denotes the $\mathcal{L}_0$ norm of the vector $\mathbf{s}$.
In fact, one can show that for a general matrix $\mathbf{\Psi}$, even finding a solution that approximates the true minimum is NP-hard \cite{WOS:000233621500010}. 
Since equation (4) is difficult to solve, problem (\ref{eq2}) can be transformed, for an approximate solution, into the following problem \cite{WOS:000278188800007,WOS:000263705000002}:
\begin{equation}\label{eq3}
	\min_{\mathbf{s}} \left \| \mathbf{s}  \right \|_{1}             ~~\text{subject to}~~ \mathbf{y} = \mathbf{\Psi s}, 
\end{equation} where $ \| \mathbf{s}\|_{1}$ denotes the $\mathcal{L}_1$ norm of the vector $\mathbf{s}$.
The design of the sensing matrix $\mathbf{\Psi}$ satisfying the restricted isometry property (RIP) \cite{WOS:000233621500010,WOS:000256409100023} is an important problem. The measurement matrix $\mathbf{A}$ satisfies the RIP of order $k$ if there exists a $\delta_{k}^{\mathbf{A}} \in (0,1)$ such that
\begin{equation}
	(1-\delta_{k}^{\mathbf{A}})\left \| \mathbf{x} \right \| _{2}^{2} \le \left \| \mathbf{Ax} \right \| _{2}^{2} \le (1+\delta_{k}^{\mathbf{A}})\left \| \mathbf{x} \right \| _{2}^{2}.
\end{equation}
The use of random matrices to construct the measurement matrix $\mathbf{A}$ offers several advantages, particularly in the context of the RIP. First, random measurement systems exhibit some robust characteristics, allowing the feasibility of the signal recovery even when a subset of measurements is corrupted or lost, provided the retained subset is sufficiently large \cite{WOS:000295420200006}. This property ensures inherent robustness to measurement errors in compressive sensing frameworks. Second, in practical scenarios where a signal $x$ is sparse under a specific transform basis $\Theta$, the critical requirement shifts to ensuring that the matrix product $\mathbf{A} \Theta$ satisfies the RIP. Unlike deterministic constructions, which necessitate explicit structural alignment between $\mathbf{A}$ and $\Theta$, random matrix designs are more general. For instance, if $\mathbf{A}$ is drawn from a Gaussian distribution and $\Theta$ is an orthonormal basis, the product $\mathbf{A}\Theta$ retains Gaussian properties, thereby satisfying the RIP with high probability. Notably, this universality extends beyond Gaussian distributions—rigorous theoretical analyses confirm analogous results for sub-Gaussian distributions and broader classes of random matrices. This universality principle underscores a key advantage of random constructions, enabling robust performance across diverse sparse representation bases without requiring prior knowledge of $\Theta$.

Recent advancements in compressed sensing have found transformative applications across diverse domains, including medical imaging systems \cite{WOS:000251346800013}, computational photography architectures \cite{WOS:000328691500013}, and distributed sensor networks \cite{WOS:000254471100013}. Deep learning methods have also been applied to some improvements of compressed sensing\cite{13}. Parallel advancements have been made in measurement matrix optimization, with notable progress in structured designs such as Toeplitz matrices \cite{WOS:000253416000061} and chaotic sensing operators \cite{WOS:000278996900005}. Notably, the semi-tensor product (STP) framework \cite{WOS:000384706700008} has emerged as a particularly effective strategy for measurement matrix enhancement, achieving dual objectives of dimensionality reduction and storage efficiency while maintaining reconstruction fidelity.
Building upon the conventional CS framework, this work proposes a novel dimension-keeping semi-tensor product (DK-STP) operation to optimize measurement matrix design. The DK-STP mechanism achieves dual objectives: (1) substantial storage reduction through dimensionality preservation, and (2) enhanced image reconstruction quality via noise suppression in compressed sensing systems. To quantitatively assess reconstruction fidelity, we employ the Peak Signal-to-Noise Ratio (PSNR) metric, where higher values indicate greater similarity between reconstructed and original images, reflecting reduced distortion. Experimental results demonstrate that DK-STP-CS achieves statistically significant PSNR improvements over both conventional CS and STP-CS baselines, validating its superior performance.

The remainder of this paper is organized as follows: Section \uppercase\expandafter{\romannumeral2} shows the specific process of compressive sensing. Section \uppercase\expandafter{\romannumeral3} presents fundamental concepts of DK-STP. Section \uppercase\expandafter{\romannumeral4} establishes the theoretical framework of DK-STP-CS, including formal proofs of its reconstruction guarantees. Section \uppercase\expandafter{\romannumeral5} conducts comprehensive simulations comparing reconstruction performance across diverse image datasets, sampling rates, and parameter configurations. Finally, Section \uppercase\expandafter{\romannumeral6} concludes with insights into future research directions.

\section{Compressed sensing problem formulation}
In some applications, a time-varying vector $\mathbf{x} \in \mathbb{R}^{n}$ needs to be transmitted (at some sampling rate) to a destination. A direct transmission would incur enormous resource consumption when the the size of $\mathbf{x}$ is large.  Therefore, it is convenient to design a constant matrix $\mathbf{A} \in \mathcal{M}_{m\times n}$ $(m\ll n)$ to generate a lower dimensional vector $\mathbf{y} \in \mathbb{R}^{m}$, with $\mathbf{y} = \mathbf{A}\mathbf{x}$, to be transmitted instead of $\mathbf{x}$. Then, at the destination, the vector $\mathbf{x}$ should be recovered from the transmitted signal $\mathbf{y}$ and the matrix $\mathbf{A}$ which is transmitted once. We refer to $\mathbf{x}$ as the input signal and $\mathbf{y}$ as the output signal. The framework consists of three key components:
\begin{itemize}
	\item \textbf{Compression}: Given a input signal $\mathbf{x} \in \mathbb{R}^{n}$, find a measurement matrix $\mathbf{A} \in \mathcal{M}_{m\times n}$ $(m\ll n)$, leading to the compressed vector  $\mathbf{y} \in \mathbb{R}^{m}$, such that 
	\[
	\mathbf{y} = \mathbf{A}\mathbf{x}.
	\]
	For non-sparse $\mathbf{x}$ with sparsity basis $\Theta$, one has $\mathbf{x}=\mathbf{\Theta}\mathbf{s}$, where $\mathbf{s}$ is a sparse signal where most components of the vector are zero. The sensing matrix is given by $\mathbf{\Psi} = \mathbf{A}\Theta$. For sparse signals $\mathbf{x}$, the sensing matrix is given by $\mathbf{\Psi} = \mathbf{A}$. Therefore,
	\[
	\mathbf{y} = \mathbf{A\Theta}\mathbf{s} = \mathbf{\Psi s}.
	\]
	
	\item \textbf{Transmission}: Transmission of the compressed signal $\mathbf{y}$ and the measurement matrix $\mathbf{A}$. 
	\item \textbf{Reconstruction}: At the arrival site, we recover the signal $\mathbf{\hat{x}}$ (an approximate of the signal $\mathbf{x}$) through optimization algorithms. We solve an approximate problem formulated as follows:
	$$\min_{\mathbf{s}} \left \| \mathbf{s}  \right \|_{1}             ~~\text{subject to}~~ \mathbf{y} = \mathbf{\Psi s},$$
	and obtain the reconstructed signal $\mathbf{\hat{x}} = \Theta \mathbf{s}$.
	
\end{itemize}
In practical applications, conventional measurement matrices often require substantial memory storage, which increases transmission costs. To address this problem, we propose a special semi-tensor product called the dimension-keeping semi-tensor product to design the measurement matrix $\mathbf{A}$. In this way, we only need to transmit a smaller measurement matrix to achieve the goal of reconstruction. This significantly reduces the bandwidth cost occupied during transmission and improves the reconstruction quality to a certain extent.
\section{Dimension-keeping semi-tensor product}
First of all, we introduce the definition of STP, which is a matrix multiplication across dimensions. When the dimension of the matrix does not meet the traditional matching conditions, it can also be calculated.
\begin{definition}\label{2.1} \cite{WOS:000255240400014}
	Let~$\mathbf{A}\in \mathcal{M}_{m\times n}$ and $\mathbf{B}\in \mathcal{M}_{p\times q}$, and $t = lcm(n,p)$ be the least common multiple of $n$ and $p$. The STP of $\mathbf{A}$ and $\mathbf{B}$, denoted by $\mathbf{A}\ltimes \mathbf{B}$, is defined as
	\begin{equation}
		\mathbf{A}\ltimes \mathbf{B} :=(\mathbf{A}\otimes \mathbf{I_{t/n}})(\mathbf{B}\otimes \mathbf{I_{t/p}}),
	\end{equation}
	where $\otimes$ is the Kronecker product, $I_{n}$ represents the $n$-order identity matrix.
\end{definition}
DK-STP is another form of STP, which is expanded from the identity matrix to a row vector with elements of $1$. Let $\mathbf{1}_n \in \mathbb{R}^n$ denote an $n$-dimensional column vector with all entries equal to one.
\begin{definition}\label{2.2}\cite{WOS:001286824800001}
	The DK-STP of $\mathbf{A}$ and $\mathbf{B}$ is expressed as follows:
	\begin{equation}
		\mathbf{A}\btimes \mathbf{B}:= (\mathbf{A}\otimes\mathbf{1_{t/n}^{T}})(\mathbf{B}\otimes\mathbf{1_{t/p}}),
	\end{equation}
	where $\mathbf{A}\in \mathcal{M}_{m\times n}, \mathbf{B}\in \mathcal{M}_{p\times q},t = lcm(n,p).$
\end{definition}

\begin{definition}\label{2.3} \cite{WOS:001286824800001}
	Let $\mathbf{A}\in \mathcal{M}_{m\times n}, \mathbf{B} \in \mathcal{M}_{p\times q}, t = lcm(n,p).$ Then the weighted DK-STP is defined by 
	\begin{equation}
		\mathbf{A} \btimes_{w}\mathbf{B} := (\mathbf{A}\otimes\bm{\varepsilon_{t/n}^{T}})(\mathbf{B}\otimes\bm{\varepsilon_{t/p}}),
	\end{equation}
	where $\bm{\varepsilon_{n}}:= \frac{1}{\sqrt{n}} \mathbf{1_{n}}, \mathbf{A} \btimes_{w}\mathbf{B} \in \mathcal{M}_{m\times q}.$
\end{definition}
In the rest of the paper, we use $\btimes$ to represent $\btimes_{w}.$
\begin{remark}
	All matrix operations in Definitions \ref{2.1}-\ref{2.3} degenerate to the traditional matrix product when the dimensionality compatibility is satisfied. So they can be regarded as an extended form of the traditional matrix product. In Definition \ref{2.3}, the number of rows of matrix $\mathbf{A}$ and the number of columns of matrix $\mathbf{B}$ determine the dimensions of the resulting product matrix. 
\end{remark}
\section{Dimension-keeping semi-tensor compressed sensing}
DK-STP enables the dimension expansion of the measurement matrix $\mathbf{A}$, thus allowing us to transmit a smaller measurement matrix $\hat{\mathbf{A}}$ compared to traditional compressed sensing while still achieving signal reconstruction. In this section, we enhance the traditional CS measurement matrix by integrating the DK-STP framework, thereby developing the DK-STP-CS model proposed in this work. 

Given a measurement matrix $\mathbf{A}\in \mathcal{M}_{m\times n}$, a input signal $\mathbf{x}\in \mathbb{R}^{p}$, a compressed signal $\mathbf{y}\in \mathbb{R}^{m}$ ($m\ll p$), we want to use DK-STP to optimize the measurement matrix $\mathbf{A}$  to achieve the compression of compressed sensing.  The following describes the process of applying the DK-STP to matrix $\mathbf{A}$ and signal $\mathbf{x}$ to derive the enhanced measurement matrix and the compressed signal $\mathbf{y}$. 

\begin{equation}
	\mathbf{y} = \mathbf{A}\btimes \mathbf{x} = (\mathbf{A}\otimes\bm{\varepsilon_{t/n}^{T}})(\mathbf{x}\otimes\bm{\varepsilon_{t/p}}), ~~~~ t = lcm(n,p).
\end{equation}In order to reduce the number of measurements, we choose $n$ to be a factor of $p$ and express its ratio as $\gamma = \frac{p}{n},$ i.e., $n|p$ and $t = p.$ Given the above assumptions, the formula can be expressed as
\begin{equation}
	\mathbf{y} = (\mathbf{A}\otimes\bm{\varepsilon_{\gamma}^{T}})\mathbf{x}.
\end{equation}When $\gamma=1$, DK-STP-CS degenerates into traditional CS. It indicates that DK-STP-CS is an extended form of CS, which can reduce the size of the measurement matrix and realize the role of compressed sensing. Compressed sensing usually processes a sparse signal $\mathbf{x}$ or writes $\mathbf{x}$ in sparse form under some sparse basis. Below we uniformly treat $\mathbf{x}$ as $k$-sparse, that is $l_{0}$ norm$$k = \left \| \mathbf{x} \right \| _{0}.$$ $k$ is just the number of elements in $\left \{i : x_{i}\ne 0  \right \} .$ If $\mathbf{x}\in \sum_{k}$, then we call $\mathbf{x}$ is $k$-sparse.

The measurement matrix needs to meet RIP conditions and incoherence, so we use $\mathbf{A}\otimes\bm{\varepsilon_{\gamma}^{T}}$ instead of $\mathbf{A}$ as the measurement matrix here. Next we will show that this operation is reasonable and give the relevant theoretical proof.
\subsection{Spark}

\begin{definition}\cite{WOS:000263705000002}
	The spark of a given matrix $\mathbf{A}$ is the smallest number
	of columns from $\mathbf{A}$ that are linearly dependent.
\end{definition}
Recall the definition in linear algebra that the largest linearly independent group of a matrix is the rank. Here we use a new definition called spark, called the minimum linearly dependent group. In this way, the uniqueness of the sparse solution can be obtained by studying spark of $\mathbf{A}$. Unfortunately, spark($\mathbf{A}$) is not easy to calculate. We typically study subsets of columns from $\mathbf{A}$.

\begin{definition}\label{3.2}
	If $\mathbf{x}\in \mathbb{R}^{p}$, then $\mathbf{x}^{\gamma}\in \mathbb{R}^{n}$, i.e., 
	\begin{equation}
		\mathbf{x}^{\gamma} := (\sum_{i=1}^{i=\gamma}x_{i} , \sum_{i=\gamma}^{i=2\gamma}x_{i},\cdots,\sum_{i=(n-1)\gamma}^{i=n\gamma}x_{i}),
	\end{equation}
	where $x_{i}$ stands for the $i$ th element in $\mathbf{x}$.
\end{definition}

The definition above means that the consecutive elements of $\mathbf{x}$ are grouped and summed to produce a new $\mathbf{x}^{\gamma}$. If $\gamma = 2$, then
\begin{equation}
	\mathbf{x}^{2} = (x_{1}+x_{2}, x_{3}+x_{4},\cdots, x_{n-1}+x_{n}).
\end{equation}

\begin{lemma}\label{1}\cite{WOS:000263705000002}
	If $k<\frac{spark(\mathbf{A})}{2}$, then for each $\mathbf{y}\in \mathbb{R}^{m}$ there exists at most one signal $\mathbf{x}\in \sum_{k}$ such that $\mathbf{y} = \mathbf{A}\mathbf{x}$.
\end{lemma}
\begin{theorem}\label{2}
	If $k<\frac{spark(\mathbf{A})}{2}$, then for each $\mathbf{y}\in \mathbb{R}^{m}$ there exists at most one signal $\mathbf{x}^{\gamma}\in \sum_{k}$ such that $\mathbf{y} = \frac{1}{\sqrt{\gamma}}\mathbf{A}\mathbf{x^{\gamma}}.$
\end{theorem}
\begin{proof}
	This is easily obtained by \textbf{Lemma} \ref{1}.
\end{proof}

The use of $\mathbf{x^{\gamma}}$ here is equivalent to treating some variable sum of the original vector as a new component. So back to our new measurement matrix which is $\mathbf{A}\otimes\bm{\varepsilon_{\gamma}^{T}}$, we find spark($\mathbf{A}\otimes\bm{\varepsilon_{\gamma}^{T}}$) = 2. This contradicts our principle of choosing a measurement matrix, we want the spark of the measurement matrix to be as large as possible. We can represent \begin{equation}\label{eq1}\small
	\mathbf{A}\otimes\bm{\varepsilon_{\gamma}^{T}}=\frac{1}{\sqrt{\gamma } } \begin{bmatrix}
		a_{11} & \cdots  & a_{11} &   \cdots & a_{1n} & \cdots & a_{1n}\\
		a_{21} & \cdots & a_{21} &    \cdots & a_{2n} & \cdots & a_{2n}\\
		\vdots  &   & \vdots &   & \vdots & &\vdots \\
		{} a_{n1} & \cdots & a_{n1}  & \cdots &a_{nn} & \cdots & a_{nn }
	\end{bmatrix}.
\end{equation}Then it is easy to verify the following equation $(\mathbf{A}\otimes\bm{\varepsilon_{\gamma}^{T}})x = \frac{1}{\sqrt{\gamma}}\mathbf{A}\mathbf{x^{\gamma}}.$
So we turned the problem of $\mathbf{y} = (\mathbf{A}\otimes\bm{\varepsilon_{\gamma}^{T}})\mathbf{x}$ into $\mathbf{y} = \frac{1}{\sqrt{\gamma}}\mathbf{A}\mathbf{x^{\gamma}}.$ This is different from what we required for spark before, now that the adjacent elements are related, we choose to deal with the sum value problem of vectors. The converted problem actually goes back to the traditional CS problem.
\begin{theorem}\label{3.5}
	If $k<\frac{spark(\mathbf{A})}{2}$, then for each $\mathbf{y}\in \mathbb{R}^{m}$ there exists one signal $\mathbf{x}\in \sum_{k}$ such that $\mathbf{y} = \mathbf{A}\btimes \mathbf{x}.$
\end{theorem}
\begin{proof}
	If $k<\frac{spark(\mathbf{A})}{2}$, then $\left \| \mathbf{x} \right \| _{0} = k.$ By \textbf{Definition} \ref{3.2}, $$\left \| \mathbf{x^{\gamma}} \right \| _{0} \le \left \| \mathbf{x} \right \| _{0} = k,$$
	hence we can get at most one signal $\mathbf{x}^{\gamma}$ by \textbf{Theorem} \ref{2}. However, $\mathbf{x^{\gamma}}$ determines the sum of the continuous components of the original signal $\mathbf{x}$, and cannot determine the specific distribution in the original signal, but for the relatively smooth vector signal, it can have a good reduction result. It will not cause too much distortion, especially for image signals. 
	
	We can obtain a signal $\mathbf{x}$ from $\mathbf{x^{\gamma}}$ by adopting some distribution principle, such as equalization. But note that this distribution is not unique. In the experiment, the adjacent values are relatively close, and the common sum values are unchanged, but it is not an accurate average distribution.
\end{proof}
\subsection{Coherence}
In traditional compressed sensing, the coherence $\mu(\mathbf{A})$ of matrix $\mathbf{A}$ is defined as$$\mu(\mathbf{A}) = \max_{1\le i \ne j \le n }  \frac{\left | \left \langle \mathbf{a_{i}},\mathbf{a_{j}} \right \rangle  \right | }{\left \| \mathbf{a_{i}} \right \|_{2} \left \| \mathbf{a_{j}} \right \|_{2}} .$$
This has been pointed out in previous studies $\mu(\mathbf{A}) \in [\sqrt{\frac{n-m}{m(n-1)} },1]$\cite{WOS:A1974T037500022}\cite{WOS:000183778200004}.
For matrix $\mathbf{A}\otimes\bm{\varepsilon_{\gamma}^{T}}$, its coherence $\mu(\mathbf{A}\otimes\bm{\varepsilon_{\gamma}^{T}}) = 1.$ But we want $\mu(\mathbf{A}\otimes\bm{\varepsilon_{\gamma}^{T}})$ to be as small as possible. We can solve the above problem by analyzing $x^{\gamma}$.

For any matrix $\mathbf{A}$, $spark(\mathbf{A}) \le 1+\frac{1}{\mu(\mathbf{A})}.$ We have the following
corollary from \textbf{Theorem} \ref{2} and \ref{3.5}.
\begin{corollary}\label{3.6}
	If $k<\frac{1}{2} (1+\frac{1}{\mu(\mathbf{A})})$, then for each $\mathbf{y}\in \mathbb{R}^{m}$ there exists at most one signal $\mathbf{x}^{\gamma}\in \sum_{k}$ such that $\mathbf{y} = \frac{1}{\sqrt{\gamma}}\mathbf{A}\mathbf{x^{\gamma}}.$
\end{corollary}

\begin{corollary}
	If $k<\frac{1}{2} (1+\frac{1}{\mu(\mathbf{A})})$, then for each $\mathbf{y}\in \mathbb{R}^{m}$ there exists one signal $\mathbf{x}\in \sum_{k}$ such that $\mathbf{y} = \mathbf{A}\btimes \mathbf{x}.$
\end{corollary}
From here we can see that the measurement matrix satisfies a certain correlation and can also achieve compressed sensing. We call it intra-group correlation. Next, the definition of intra-group correlation is given.

\begin{definition}\label{3.8}
	If the columns of matrix $\mathbf{A}$ are equally divided into $n$ groups which have the same column vectors, and these vectors are incoherent to each other, then we call this intra-group correlation.
\end{definition}
To clarify the above definition, we give a Gaussian random matrix $\mathbf{A}$, which is a matrix with columns incoherence. Then $\mathbf{A}\otimes\bm{\varepsilon_{\gamma}^{T}}$ is a matrix that satisfies the \textbf{Definition} \ref{3.8}, and it can be inferred that there is still a unique solution for $\mathbf{x^{\gamma}}$ by \textbf{Corollary} \ref{3.6}. This indicates that after the column grouping of the measurement matrix, if the vectors between groups remain uncorrelated while those within groups are correlated, a very good reconstruction effect can still be achieved in terms of images.

When there are few vectors in the group, this correlation actually does not affect us to do compressed sensing, especially in image compressed sensing. This part will be discussed in the following experiment.
\subsection{Restricted isometry property}

\begin{lemma}\label{7}\cite{WOS:000384706700008}
	For any $p$-dimensional vector $\mathbf{x} \in \sum_{k}$, if the sensing matrix $\mathbf{A}$ satisfies the RIP of order $2k$ with the constant $\delta_{2k}^{\mathbf{A}} < \sqrt{2} - 1 $, then the exact recovery in CS model is possible by the $l_{1}$ optimization.
\end{lemma}
The theorem above shows us that if the sensing matrix meets the RIP condition, the precise recovery of compressed sensing can be achieved through $l_{1}$ optimization. 
\begin{theorem}
	If $\mathbf{A}$ satisfies the RIP of order $2k$ with constant $\delta_{2k}^{\mathbf{A}} < \sqrt{2} - 1 $, then the exact recovery for $\mathbf{y} = \frac{1}{\sqrt{\gamma}}\mathbf{A}\mathbf{x^{\gamma}}$ is possible by the $l_{1}$ optimization.
\end{theorem}
\begin{proof}
	If A satisfies the RIP of order $2k$ with constant $\delta_{2k}^{\mathbf{A}}$, then $$(1-\delta_{2k}^{\mathbf{A}})\left \| \mathbf{x} \right \| _{2}^{2} \le \left \| \mathbf{Ax} \right \| _{2}^{2} \le (1+\delta_{2k}^{\mathbf{A}})\left \| \mathbf{x} \right \| _{2}^{2}.$$
	Take $\mathbf{x} = \frac{1}{\sqrt{\gamma}}\mathbf{x}^{\gamma}$ , we find $$(1-\delta_{2k}^{\mathbf{A}})\left \| \mathbf{x}^{\gamma} \right \| _{2}^{2} \le \left \| \mathbf{A\mathbf{x}^{\gamma}} \right \| _{2}^{2} \le (1+\delta_{2k}^{\mathbf{A}})\left \| \mathbf{x}^{\gamma} \right \| _{2}^{2}.$$
	By \textbf{Lemma} \ref{7}, the theorem has been proved.
\end{proof}
\subsection{Error analysis}
Reconstruction error analysis is an important part, as it is directly related to the reconstruction performance and serves as a key indicator of the effectiveness of the method. This subsection analyzes the reconstruction error in the proposed compressed sensing based dimension-keeping semi-tensor product  method, which is attributed to three fundamental sources: the original signal error, the compressed sensing error, and the distribution error, thereby providing a fundamental explanation of the error causes and highlighting the superiority of our method. To this end, some symbol definitions are provided.
\begin{itemize}
	\item $x^{*}$ represents the reconstructed signal for original signal $\mathbf{x}$.
	\item $x_{i}$ represents each component of the signal $\mathbf{x}$.
	\item $x_{i}^{*}$ represents each component of the signal $\mathbf{x^{*}}$.
	\item $\mathbf{x^{\gamma *}}$ represents the sum of each component of the reconstructed signal $\mathbf{x^{*}}$.
	$$\mathbf{x}^{\gamma *} := (\sum_{i=1}^{i=\gamma}x_{i}^{*} , \sum_{i=\gamma}^{i=2\gamma}x_{i}^{*},\cdots,\sum_{i=(n-1)\gamma}^{i=n\gamma}x_{i}^{*}).$$
	\item $\frac{\mathbf{x^{\gamma }}}{\gamma}$ and $\frac{\mathbf{x^{\gamma *}}}{\gamma}$ represent the average number. Their dimensions are different from the original signal $\mathbf{x}$, so each component is repeated $\gamma$ times to obtain $\mathbf{\bar{x}}$ and $\mathbf{\bar{x^{*}}}$ with the same dimension as the original signal $x$. Such as $$\mathbf{\bar{x}}:= (\underbrace{\sum_{i=1}^{i=\gamma}\frac{x_{i}}{\gamma} ,\sum_{i=1}^{i=\gamma}\frac{x_{i}}{\gamma},\cdots,\sum_{i=1}^{i=\gamma}\frac{x_{i}}{\gamma}}_{\gamma}, \sum_{i=\gamma}^{i=2\gamma}\frac{x_{i}}{\gamma},\cdots,\sum_{i=(n-1)\gamma}^{i=n\gamma}\frac{x_{i}}{\gamma}).$$
\end{itemize} 
 
\begin{theorem}\label{4.11}
	For the estimation of $\left \|\mathbf{x^{*}} - \mathbf{x}  \right \|_{2}$, the following norm inequality holds$$\left \| \mathbf{x^{*}} - \mathbf{x} \right\|_{2} \le \left \|\mathbf{x^{*}}-\mathbf{\bar{x}} \right \|_{1} + \left \|\mathbf{x^{\gamma}}-\mathbf{x^{\gamma *}}\right \|_{1} +\left \|\mathbf{\bar{x}}- \mathbf{x} \right \|_{1}.$$
\end{theorem}
\begin{proof}
	Note that
	\begin{align*}
			\left \| \mathbf{x^{*}} - \mathbf{x} \right\|_{1} &= \left \| \mathbf{x^{*}} - \mathbf{\bar{x^{*}}} + \mathbf{\bar{x^{*}}} - \mathbf{\bar{x}}+\mathbf{\bar{x}}-\mathbf{x} \right\|_{1}\\
			&\le \left \|\mathbf{x^{*}} - \mathbf{\bar{x^{*}}} \right \|_{1} + \left \|\mathbf{\bar{x^{*}}} - \mathbf{\bar{x}} \right \|_{1} + \left \|\mathbf{\bar{x}}-\mathbf{x} \right\|_{1}\\
			&=\left \|\mathbf{x^{*}} - \mathbf{\bar{x}}+\mathbf{\bar{x}}-\mathbf{\bar{x^{*}}} \right \|_{1} + \left \|\mathbf{\bar{x^{*}}} - \mathbf{\bar{x}} \right \|_{1} + \left \|\mathbf{\bar{x}}-\mathbf{x} \right\|_{1}\\
			&\le \left \|\mathbf{x^{*}} - \mathbf{\bar{x}} \right \|_{1} + 2\left \|\mathbf{\bar{x}} - \mathbf{\bar{x^{*}}} \right \|_{1}+\left \|\mathbf{\bar{x}}-\mathbf{x} \right\|_{1}\\
			&=  \left \|\mathbf{x^{*}} - \mathbf{\bar{x}} \right \|_{1}  +  \left \|\mathbf{x^{\gamma}}-\mathbf{x^{\gamma *}}\right \|_{1} + \left \|\mathbf{\bar{x}}-\mathbf{x} \right\|_{1},
	\end{align*}
then it follows that
	$$\left \| \mathbf{x^{*}} - \mathbf{x} \right\|_{2} \le 	\left \| \mathbf{x^{*}} - \mathbf{x} \right\|_{1} \le \left \|\mathbf{x^{*}} - \mathbf{\bar{x}} \right \|_{1}  +  \left \|\mathbf{x^{\gamma}}-\mathbf{x^{\gamma *}}\right \|_{1} + \left \|\mathbf{\bar{x}}-\mathbf{x} \right\|_{1}.$$

\end{proof}

From Theorem \ref{4.11}, the $l_{2}$ norm of this reconstruction error is restricted by three items: \begin{itemize}
	\item Distribution error $\left \|\mathbf{x^{*}} - \mathbf{\bar{x}} \right \|_{1}.$
	\item Compressed sensing error $\left \|\mathbf{x^{\gamma}}-\mathbf{x^{\gamma *}}\right \|_{1}.$
	\item Original signal error $\left \|\mathbf{\bar{x}}-\mathbf{x} \right\|_{1}.$
\end{itemize}
Below, we will conduct data analysis on these three types of errors respectively through three experimental images.
\subsubsection{Original signal error}
For the original signal error, the element values in the signal should not deviate too much from the mean of their adjacent elements. This can ensure that the error is as small as possible. In simple terms, the signal needs to be as continuous as possible and avoid sudden and excessive fluctuations. For common signal sources, the relative continuity of image signals is relatively strong, and the difference in adjacent pixel values is not significant, which is quite suitable for us to select the original signal here. So next, we respectively select the truncated parts from different images as the original signals to calculate the magnitude of this error.

In image processing, the $l_{1}$ norm error is usually used to calculate the MAE value for comparison. The MAE is defined as follows:$$\mathrm{MAE} = \frac{1}{N}\sum_{i=1}^{N}|X_i - Y_i|,$$ where $X_{i}$ represents the $i$-th element value of the original signal $X$, $Y_{i}$ represents the $i$-th element value of the reconstructed signal $Y$. 

The calculation of the mean signal $\bar{x}$ for three natural images is called the reconstructed image. Calculate its heat map of error distribution to observe the overall error distribution of the image, and draw its frequency distribution histogram to help analyze the error. The above computational plots were respectively performed on the three images $Pepper$, $Baboon$ and $Barbara$ to obtain Fig. \ref{original error2}-\ref{original error4}.
\begin{figure}[ht]
	\centering
	
	\subfloat[]{\includegraphics[width=0.3\columnwidth]{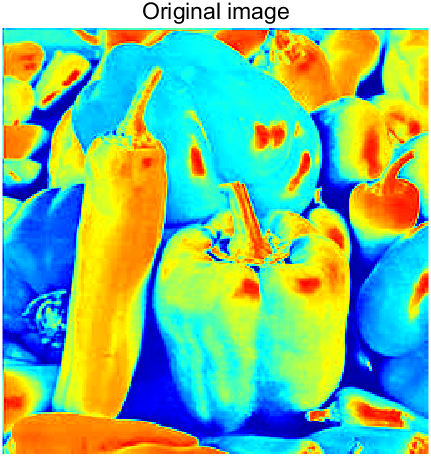}\label{image2ori}}
	\hfill
	\subfloat[]{\includegraphics[width=0.3\columnwidth]{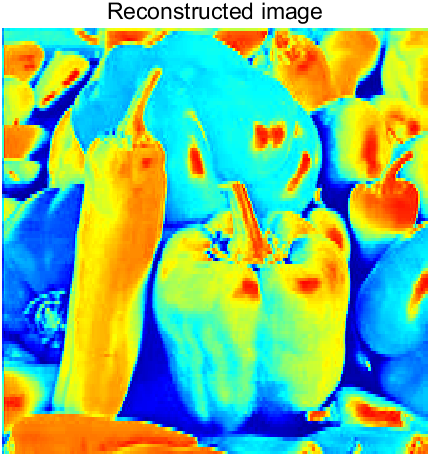}\label{image2rec}}
	\hfill
	\subfloat[]{\includegraphics[width=0.36\columnwidth]{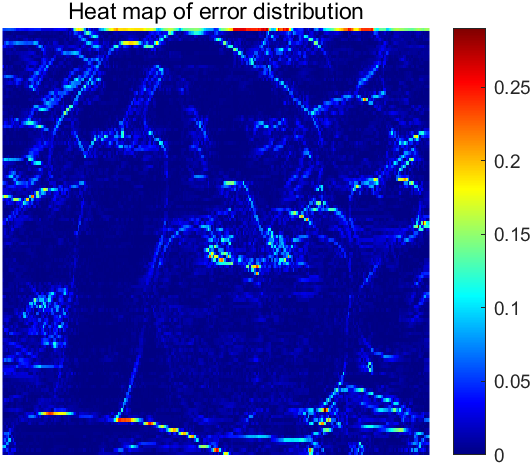}\label{image2heat}}
	
	\subfloat[]{\includegraphics[width=0.55\columnwidth]{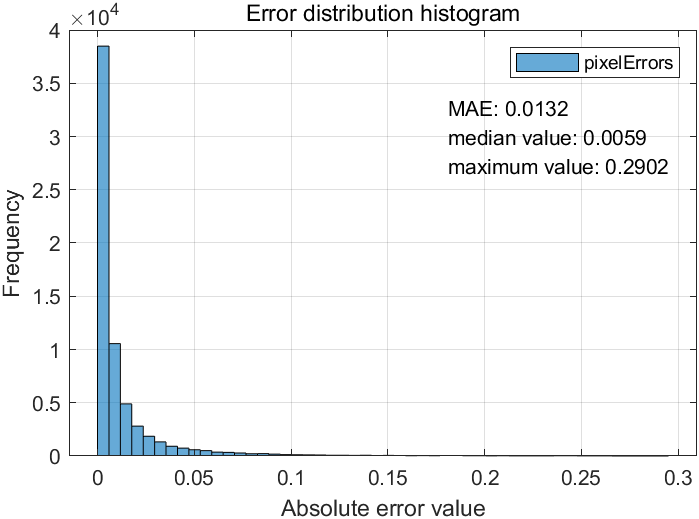}\label{image2hist}}
	\caption{(a) represents the original heat map of $Pepper$, (b) represents $Pepper$'s mean reconstruction image, (c) represents the heat map representation of the error image between the original Pepper image and the reconstructed image, (d) represents the histogram of the error frequency distribution of different pixels.}
	\label{original error2}
\end{figure}
\begin{figure}[ht]
	\centering
	\subfloat[]{\includegraphics[width=0.3\columnwidth]{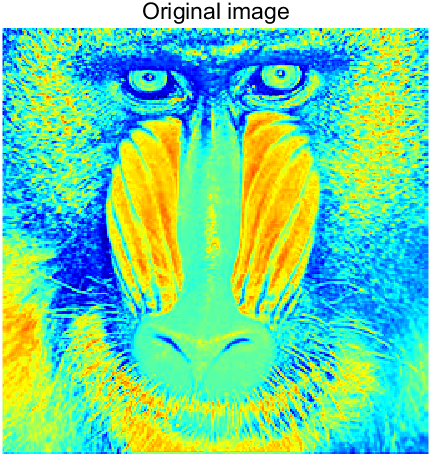}\label{image3ori}}
	\hfill
	\subfloat[]{\includegraphics[width=0.3\columnwidth]{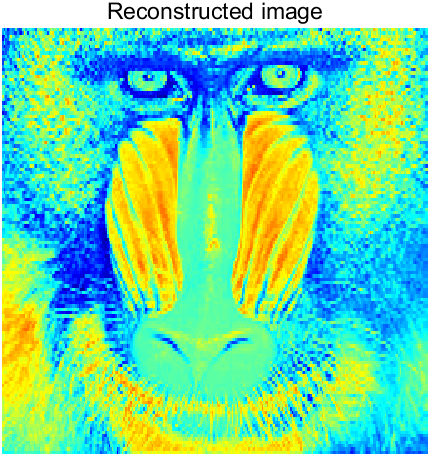}\label{image3rec}}
    \hfill
	\subfloat[]{\includegraphics[width=0.36\columnwidth]{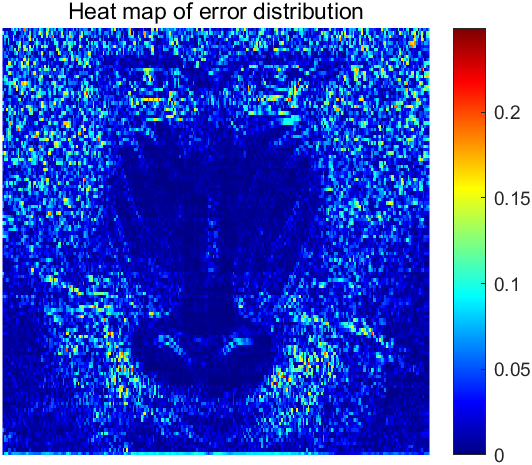}\label{image3heat}}
	
	\subfloat[]{\includegraphics[width=0.55\columnwidth]{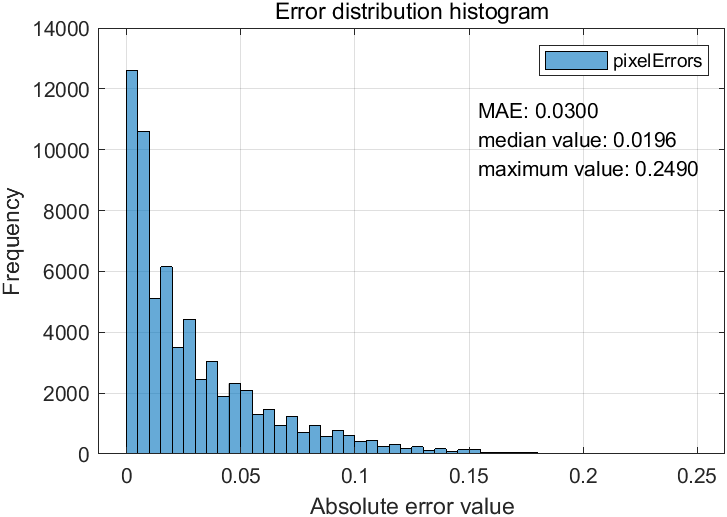}\label{image3hist}}

	\caption{(a) represents the original heat map of $Baboon$, (b) represents $Baboon$'s mean reconstruction image, (c) represents the heat map representation of the error image between the original Pepper image and the reconstructed image, (d) represents the histogram of the error frequency distribution of different pixels.}
	\label{original error3}
\end{figure}
\begin{figure}[ht]
	\centering
	\subfloat[]{\includegraphics[width=0.3\columnwidth]{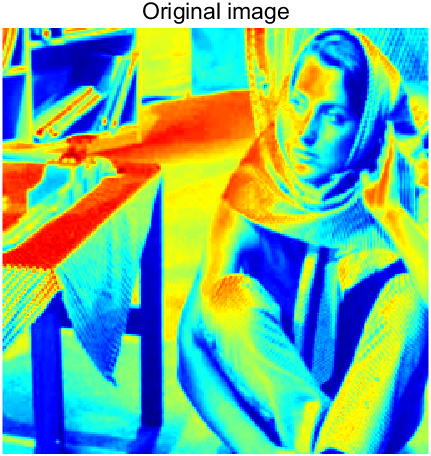}\label{image4ori}}
	\hfill
	\subfloat[]{\includegraphics[width=0.3\columnwidth]{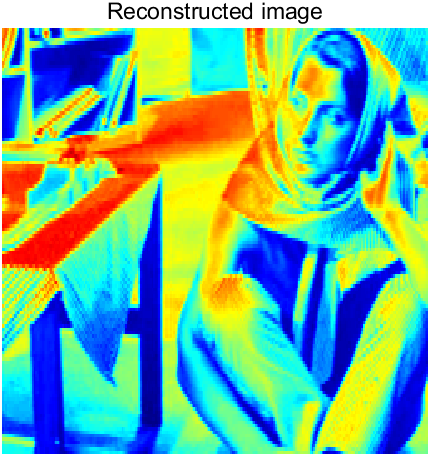}\label{image4rec}}
	\hfill
	\subfloat[]{\includegraphics[width=0.36\columnwidth]{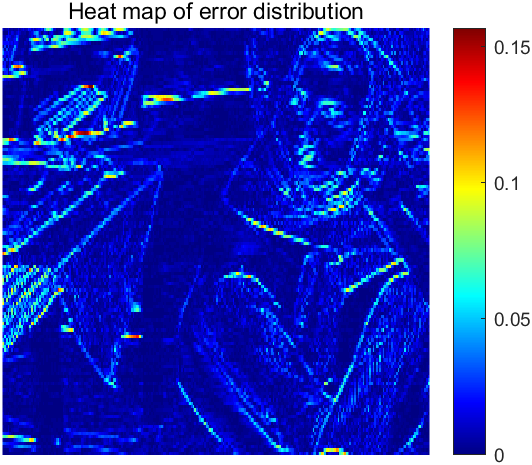}\label{image4heat}}
	
	\subfloat[]{\includegraphics[width=0.55\columnwidth]{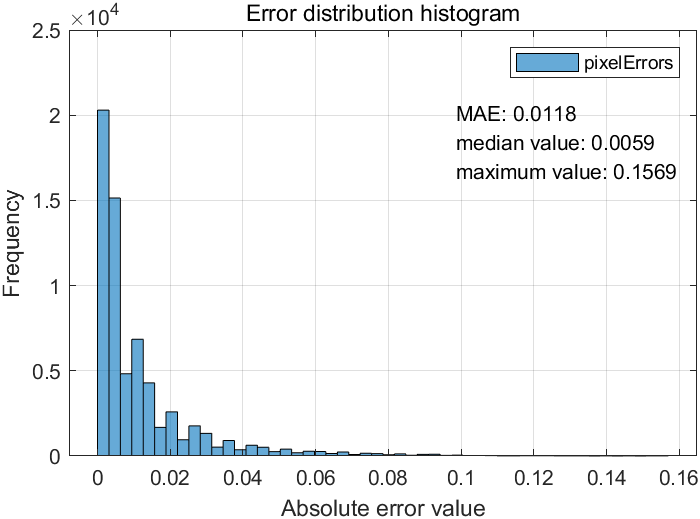}\label{image4hist}}

	\caption{(a) represents the original heat map of $Barbara$, (b) represents $Barbara$'s mean reconstruction image, (c) represents the heat map representation of the error image between the original Pepper image and the reconstructed image, (d) represents the histogram of the error frequency distribution of different pixels.}
	\label{original error4}
\end{figure}
It can be found from the three experimental results that the original signal error of the natural image is relatively small, and the places with larger errors occur at the boundaries, which account for a relatively small proportion of the overall error. It can be seen from the frequency distribution histogram that the errors of most pixels are concentrated around 0, which ensures that the error of the original signal is relatively small.
\subsubsection{Compressed sensing error}
The compressive sensing error here represents the error of the signal composed of the sum values of adjacent pixels, which indicates that we have shifted from the reconstruction of the original signal to the reconstruction of signals with fewer sum values. To some extent, this reduces the compression ratio, which will make the reconstruction of the sum value signal more accurate. The reduced error here is greater than the two types of increased errors, so it makes the effect of image reconstruction better. In Fig. \ref{compressed error2}-\ref{compressed error4}, five 64$\times$64 pixel blocks are randomly selected from different images for compressive sensing. The compression ratio ranges from 5\% to 100\%, increasing by 5\% each time. The average MAE value of these five pixel blocks is calculated, and the relevant images are drawn. Taking $\gamma$ equal to 2 as an example, we calculate the change in MAE value when the compression ratio varies by twice and draw a histogram.
\begin{figure}[ht]
	\centering
	\subfloat[]{\includegraphics[width=0.35\columnwidth]{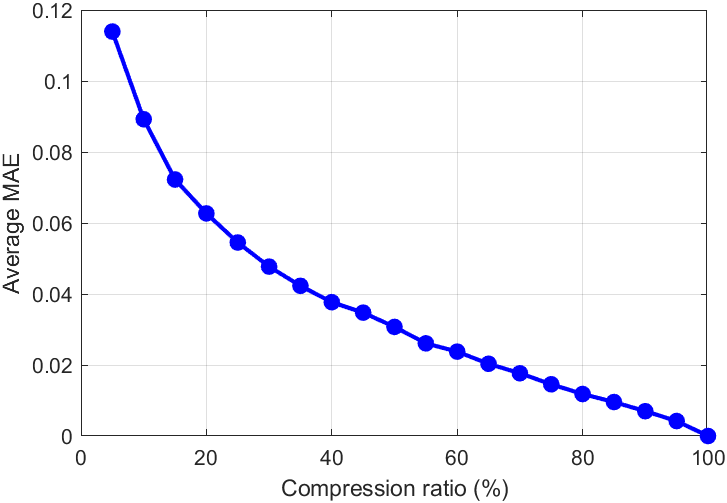}\label{image2mae}}
	\hfill
	\subfloat[]{\includegraphics[width=0.35\columnwidth]{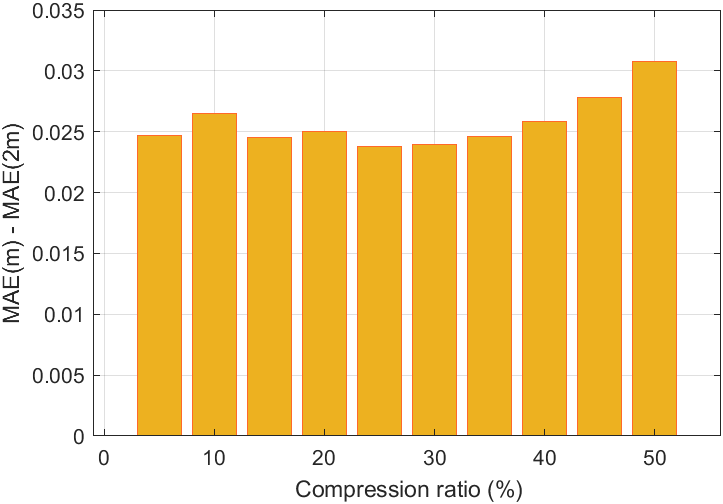}\label{image2maed}}
	\hfill
	\subfloat[]{\includegraphics[width=0.25\columnwidth]{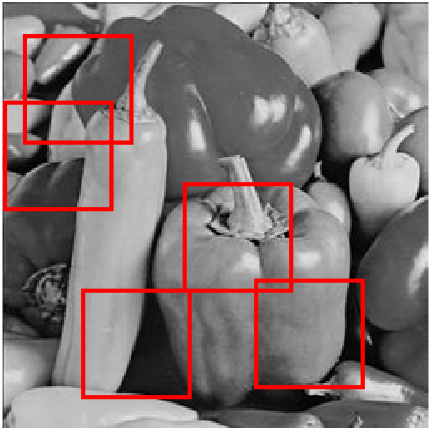}\label{image2pos}}

	\caption{(a) represents the relationship between the MAE value of $Pepper$ and its compression ratio, (b) represents the MAE difference under a compression ratio difference of twice, (c) represents the selected 5 pixel blocks.}
	\label{compressed error2}
\end{figure}

\begin{figure}[ht]
	\centering
	\subfloat[]{\includegraphics[width=0.35\columnwidth]{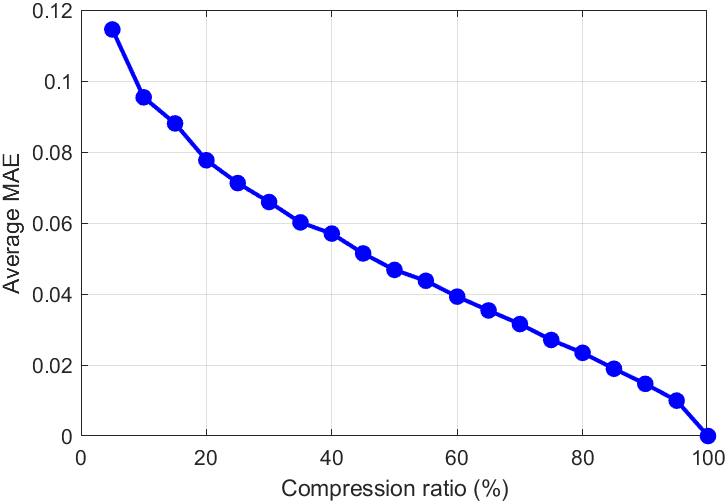}\label{image3mae}}
	\hfill
	\subfloat[]{\includegraphics[width=0.35\columnwidth]{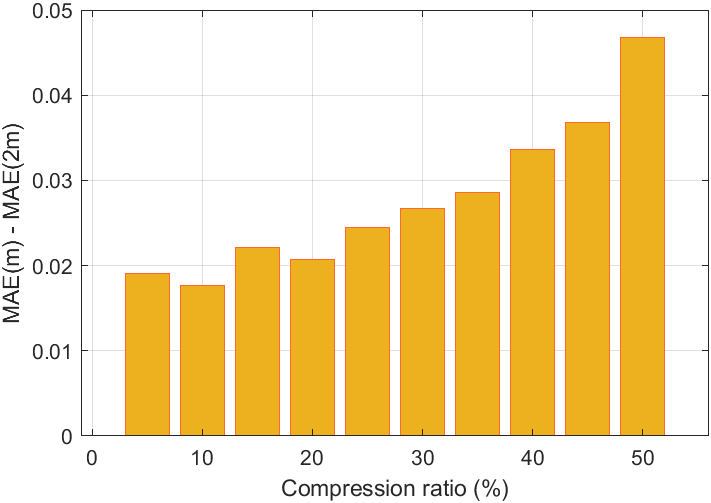}\label{image3maed}}
	\hfill
	\subfloat[]{\includegraphics[width=0.25\columnwidth]{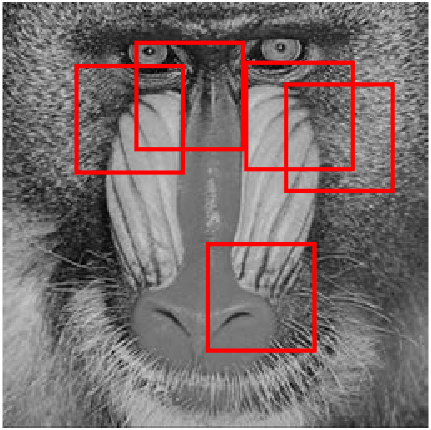}\label{image3pos}}

	\caption{(a) represents the relationship between the MAE value of $Baboon$ and its compression ratio, (b) represents the MAE difference under a compression ratio difference of twice, (c) represents the selected 5 pixel blocks.}
	\label{compressed error3}
\end{figure}

\begin{figure}[ht]
	\centering
	\subfloat[]{\includegraphics[width=0.35\columnwidth]{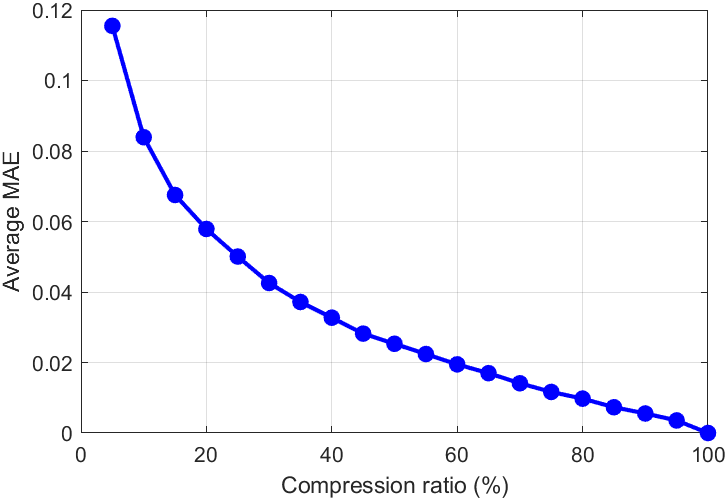}\label{image4mae}}
	\hfill
	\subfloat[]{\includegraphics[width=0.35\columnwidth]{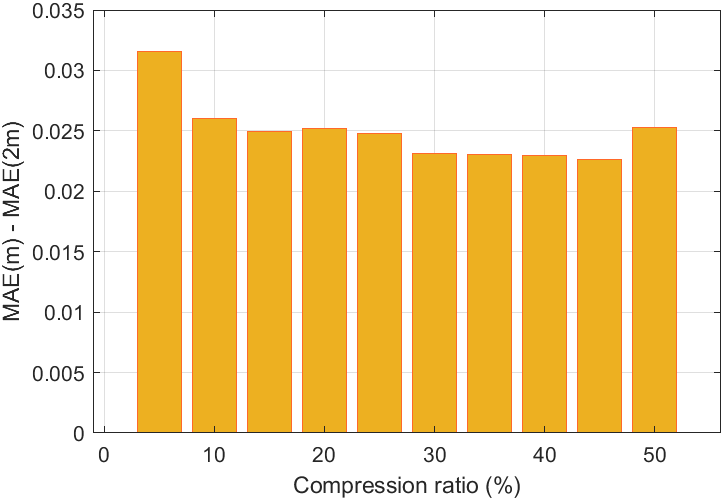}\label{image4maed}}
	\hfill
	\subfloat[]{\includegraphics[width=0.25\columnwidth]{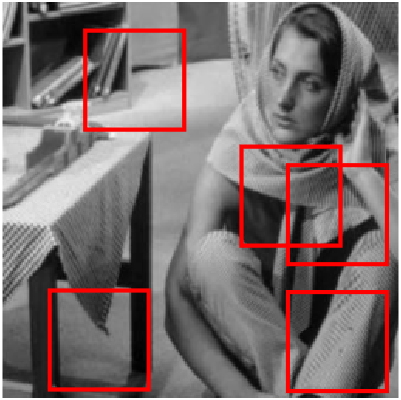}\label{image4pos}}

	\caption{(a) represents the relationship between the MAE value of $Barbara$ and its compression ratio, (b) represents the MAE difference under a compression ratio difference of twice, (c) represents the selected 5 pixel blocks.}
	\label{compressed error4}
\end{figure}

It can be seen from the image that when the compression ratio increases by two times, the MAE value has significantly improved. This is also the reason why in DK-STP-CS, the compression ratio of the sum value has been increased while the overall compression ratio remains unchanged. The overall quality of the image was retained at the expense of some details. 

\subsubsection{Distribution error}
The essence of the allocation error is that we only recover the sum value signal, but the information on how to allocate the sum value to each element is lost. The loss of local information has made the overall information more accurate. This part of the error is also limited by the error of the original signal. If the error of the original signal of a signal is 0, then we can distribute it evenly. In this way, the distribution error will not occur.
In natural images, distribution errors are bound to exist and to some extent depend on the pixel values of the image and the compression ratio. Below, we present Fig. \ref{distribution error} to represent the distribution error curves of the three images.

\begin{figure}[ht]
	\centering
	\subfloat[]{\includegraphics[width=0.33\columnwidth]{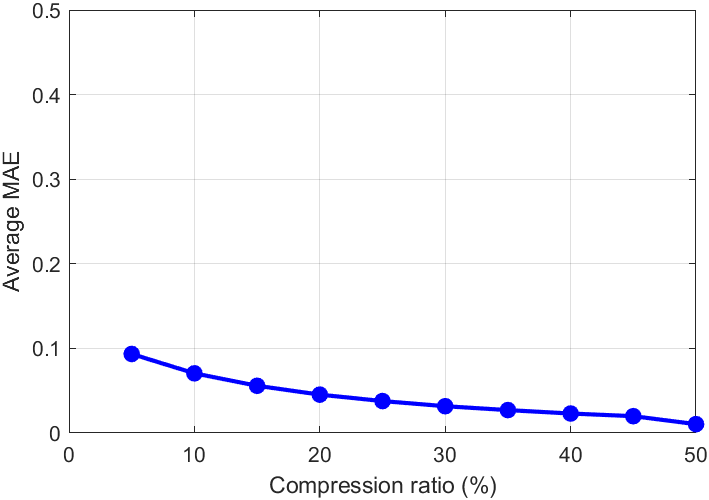}\label{image2dis}}
	\hfill
	\subfloat[]{\includegraphics[width=0.33\columnwidth]{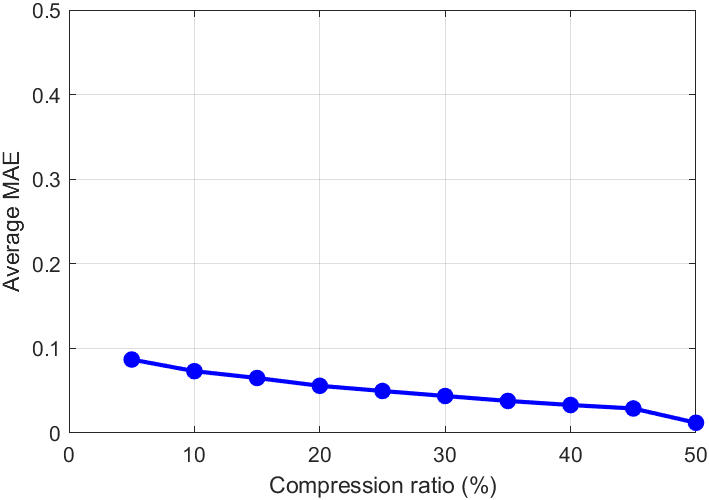}\label{image3dis}}
	\hfill
	\subfloat[]{\includegraphics[width=0.33\columnwidth]{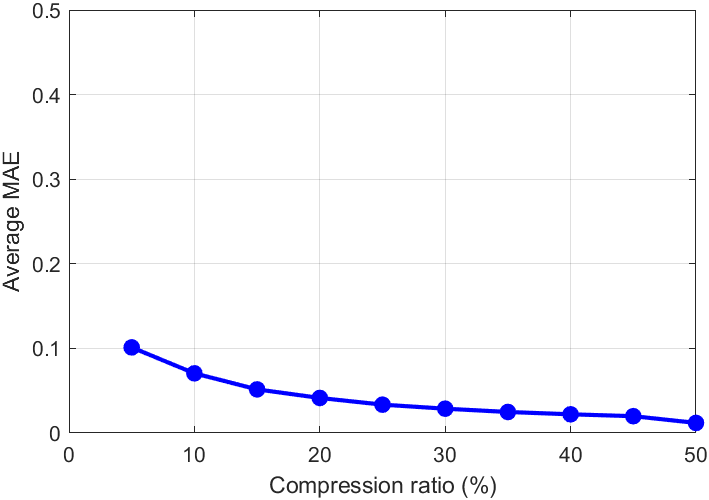}\label{image4dis}}
	
	\caption{(a), (b), (c) respectively represent the distribution error curves of $Pepper$, $Baboon$ and $Barbara$. The vertical axis represents the average MAE value and the horizontal axis represents the compression ratio.}
	\label{distribution error}
\end{figure}
In terms of overall error analysis, our approach can effectively enhance the overall quality of the image while neglecting some detailed information. Under the premise that the image is large enough, these lost details can be ignored. However, the overall improvement in quality can be clearly felt.

The theoretical foundation of our reconstruction framework is established through rigorous analysis of two critical matrix properties: \textit{Spark} and \textit{Coherence}. These metrics collectively ensure the uniqueness of sparse solutions and validate the feasibility of the proposed approach. Specifically, the Spark parameter guarantees solution uniqueness under sparsity constraints, while Coherence quantification governs measurement matrix optimization. Furthermore, the algorithmic implementation is formally presented in Algorithm~\ref{dk-stp-cs}, which outlines the complete reconstruction pipeline.
\begin{algorithm}[H]
	\caption{DK-STP-CS Reconstruction Algorithm}
	\label{dk-stp-cs}
	\begin{algorithmic}[1]
		\REQUIRE Measurement signal $\mathbf{y} \in \mathbb{R}^m$, sensing matrix $\mathbf{A} \in \mathbb{R}^{m \times \frac{p}{\gamma}}$, parameter $\gamma$
		\ENSURE Reconstructed image $\mathbf{x} \in \mathbb{R}^{p}$
		\FOR{$i = 1$ to $\gamma$}
		\STATE \quad $x^i \gets \left(x_{(i-1)\gamma+1}+x_{(i-1)\gamma+2}+\cdots+x_{i\gamma}\right)$
		\ENDFOR
		\STATE \quad $\mathbf{X} \gets \begin{bmatrix}x^1 & x^2 & \cdots & x^\gamma\end{bmatrix} \in \mathbb{R}^{\frac{p}{\gamma}}$
		\STATE $\mathbf{y} = \hat{\mathbf{A}}\mathbf{x} = (\mathbf{A}\otimes\bm{\varepsilon_{\gamma}^{T}})\mathbf{x} = \frac{1}{\sqrt{\gamma}}\mathbf{A}\mathbf{X}$
		\STATE Solve problem:
		\STATE \quad $\mathbf{X} \gets \arg\min_{\mathbf{z}} \|\mathbf{z}\|_1 \ \text{s.t.} \ \mathbf{A}\mathbf{z} = \mathbf{y}$
		\STATE Reshape results: 
		\STATE $\mathbf{x} \gets \mathbf{X}$
		\RETURN $\mathbf{x}$
	\end{algorithmic}
\end{algorithm}
\section{Experiment}
 To assess the reconstruction efficacy of the DK-STP-CS framework, we conducted comparative analyses against standard CS and STP-CS methodologies, employing Basis Pursuit (BP) as the unified recovery algorithm. Specifically, the performance differentials among DK-STP-CS, STP-CS, and CS were rigorously quantified. Furthermore, we analyzed the time complexity and storage requirements of the DK-STP-CS architecture to evaluate its computational feasibility.

In subsequent experimental validations, the sensing matrix for DK-STP-CS was instantiated as a Gaussian random matrix. The evaluation protocol incorporated visual quality assessments and quantitative metrics, with a focus on image signal recovery tasks. The reconstruction fidelity was measured via the Peak Signal-to-Noise Ratio (PSNR), defined as:$$PSNR = 10lg[\frac{(MAX_{I})^{2}}{MSE}],$$in this context, $\xi$ denotes the bit depth per pixel, conventionally set to $\xi = 8$ for standard grayscale images, corresponding to a maximum intensity value of $MAX_{I} =2^{\xi}-1= 255$. During subsequent experimental procedures, the PSNR is expressed in dB. The mean squared error (MSE) quantifies the discrepancy between the original image $I$ and the reconstructed image $J$, formally defined as:

$$MSE = \frac{1}{mn} \sum_{i = 0}^{m-1} \sum_{j=0}^{n-1}(I(i,j) - J(i,j))^{2}, $$where $m\times n$ represents the image dimensions, and $I(i,j), J(i,j)$ denote pixel intensities at position $(i,j)$ in the original and recovered images, respectively. This metric provides a pixel-wise fidelity assessment critical for evaluating reconstruction accuracy in imaging applications. We take the grayscale image and transform it into a matrix where the elements represent the grayscale value of each pixel cell. That is, we correspond each image to a matrix, and arrange each column of the matrix in a column in order to obtain a vector signal. 

Given a grayscale image represented as matrix $\mathbf{A} \in \mathcal{M}_{m \times n}$, the corresponding vector representation $\mathbf{x} \in \mathbb{R}^{mn}$ can be obtained through column-wise vectorization, expressed as:
\[\mathbf{x} = \text{vec}(\mathbf{A}) = [a_{11},...,a_{m1},a_{12},...,a_{m2},...,a_{1n},...,a_{mn}]^T,\]
where $\text{vec}(\cdot)$ denotes the vectorization operator that stacks matrix columns sequentially.

\subsection{Visual and PSNR comparisons of different images at the same sampling rate.}
The proposed framework is validated using a standard $256\times256$ grayscale test image, which is vectorized into $\mathbf{x} \in \mathbb{R}^{65536}$ through column-wise concatenation. Three distinct measurement matrices are implemented for comparative analysis:

\begin{itemize}
	\item \textbf{Conventional CS}: $\mathbf{A}_{CS} \in \mathbb{R}^{m \times n}$ with i.i.d. Gaussian entries $\mathcal{N}(0,1),$
	\item \textbf{STP-CS}: $\mathbf{A}_{STP-CS} = \mathbf{A}_{\frac{m}{\gamma} \times \frac{n}{\gamma}} \otimes \mathbf{I}_\gamma$, where $\otimes$ denotes Kronecker product,
	\item \textbf{DK-STP-CS}: $\mathbf{A}_{DK-STP-CS} = \mathbf{A}_{m \times \frac{n}{\gamma}} \otimes \bm{\varepsilon}_\gamma^T$, with $\bm{\varepsilon}_\gamma = \frac{1}{\sqrt{\gamma}}\mathbf{1}_\gamma,$
\end{itemize}

where $n=65536$ (original dimension), $m=32768$ (compressed dimension), $\gamma = 2$ (parameter). The signal sparsity is induced through discrete cosine transform (DCT) basis $\Theta$, such that $\mathbf{x} = \Theta\mathbf{s}$. This configuration yields a compression ratio $CR = m/n = 0.5$, enabling systematic evaluation of reconstruction fidelity under controlled dimensionality reduction.

\begin{figure}[ht]
	\centering
	\subfloat[]{\includegraphics[width=0.22\columnwidth]{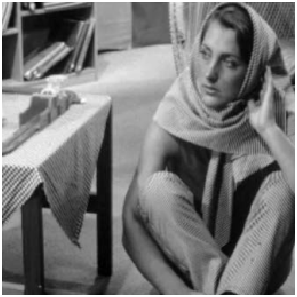}\label{fig:barbara_orig}}
	\hfill
	\subfloat[]{\includegraphics[width=0.22\columnwidth]{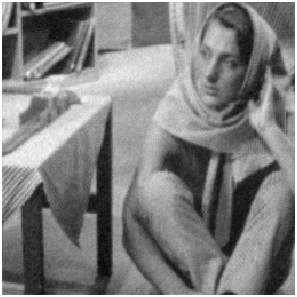}\label{fig:barbara_cs}}
	\hfill
	\subfloat[]{\includegraphics[width=0.22\columnwidth]{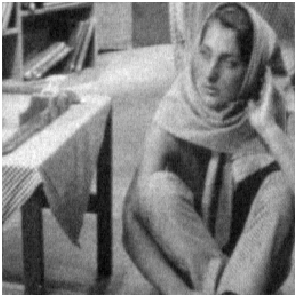}\label{fig:barbara_stpcs}}
	\hfill
	\subfloat[]{\includegraphics[width=0.22\columnwidth]{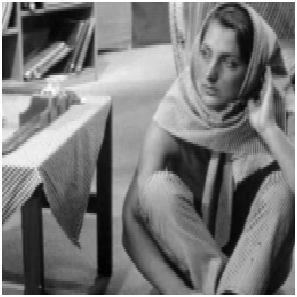}\label{fig:barbara_dkstp}}
	
	\subfloat[]{\includegraphics[width=0.22\columnwidth]{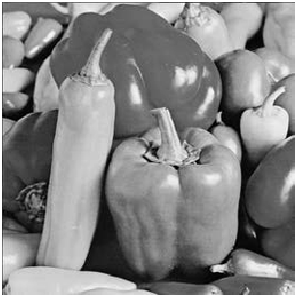}\label{fig:pepper_orig}}
	\hfill
	\subfloat[]{\includegraphics[width=0.22\columnwidth]{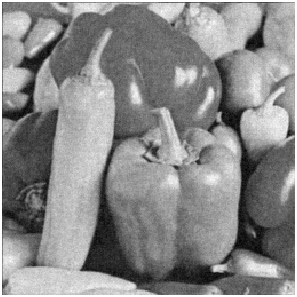}\label{fig:pepper_cs}}
	\hfill
	\subfloat[]{\includegraphics[width=0.22\columnwidth]{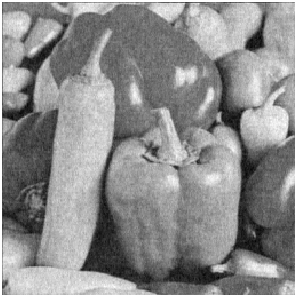}\label{fig:pepper_stpcs}}
	\hfill
	\subfloat[]{\includegraphics[width=0.22\columnwidth]{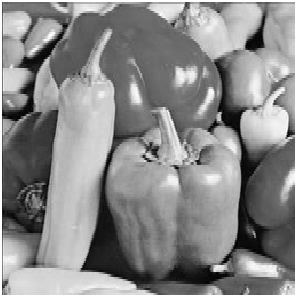}\label{fig:pepper_dkstp}}
	
	\subfloat[]{\includegraphics[width=0.22\columnwidth]{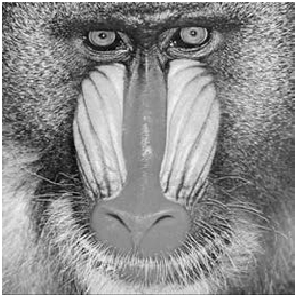}\label{fig:baboon_orig}}
	\hfill
	\subfloat[]{\includegraphics[width=0.22\columnwidth]{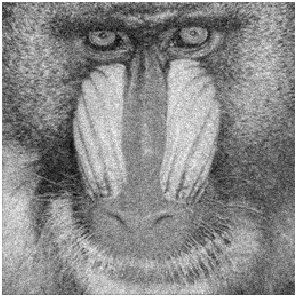}\label{fig:baboon_cs}}
	\hfill
	\subfloat[]{\includegraphics[width=0.22\columnwidth]{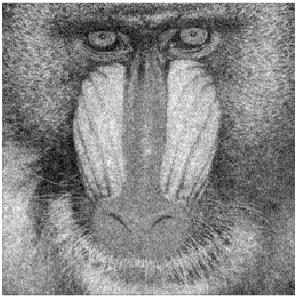}\label{fig:baboon_stpcs}}
	\hfill
	\subfloat[]{\includegraphics[width=0.22\columnwidth]{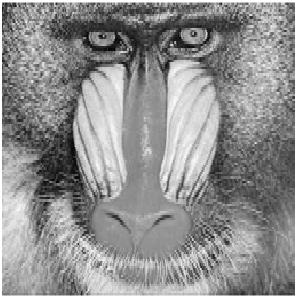}\label{fig:baboon_dkstp}}
	
	\caption{Recovered images by different CS models. (a), (e) and (i) represent the original $Barbara$, $Pepper$, and $Baboon$ images (256$\times$256) respectively. (b), (f) and (j) CS method. (c), (g) and (k) STP-CS method. (d), (h) and (l) DK-STP-CS method. In the measurement process, the Gaussian random matrix is taken as the measurement matrix. The recovery method is BP.}
	\label{figure1}
\end{figure}

As illustrated in Figure \ref{figure1} under a compression ratio (CR) of 0.5, which corresponds to sampling only 50\% of the original pixels, the reconstructed images exhibit distinct performance variations across methodologies. Visual inspection reveals that conventional CS and STP-CS frameworks introduce substantial noise artifacts, significantly impairing image recognition fidelity. In contrast, the proposed DK-STP-CS architecture effectively suppresses noise contamination, as evidenced by smoother texture regions and reduced granular distortions. This improvement aligns with the inherent mechanism of DK-STP-CS: during reconstruction, adjacent pixel blocks are restored holistically through the dimension-keeping semi-tensor product operation, thereby mitigating abrupt intensity fluctuations between neighboring regions. However, edge preservation remains a relative limitation of DK-STP-CS. Quantitative analysis of Figure \ref{figure1}(d,h,l) demonstrates that while the method achieves superior intra-block coherence, it occasionally introduces mild aliasing artifacts along high-frequency edges compared to CS and STP-CS baselines. This trade-off between noise suppression and edge sharpness reflects the inter-group correlation characteristics embedded in the DK-STP measurement matrix design, as theoretically established in Section 3.2. The observed performance dichotomy underscores the method's suitability for applications prioritizing homogeneous region fidelity over fine edge reconstruction.

\begin{figure}[ht]
		\centering
	\subfloat[]{\includegraphics[width=0.31\columnwidth]{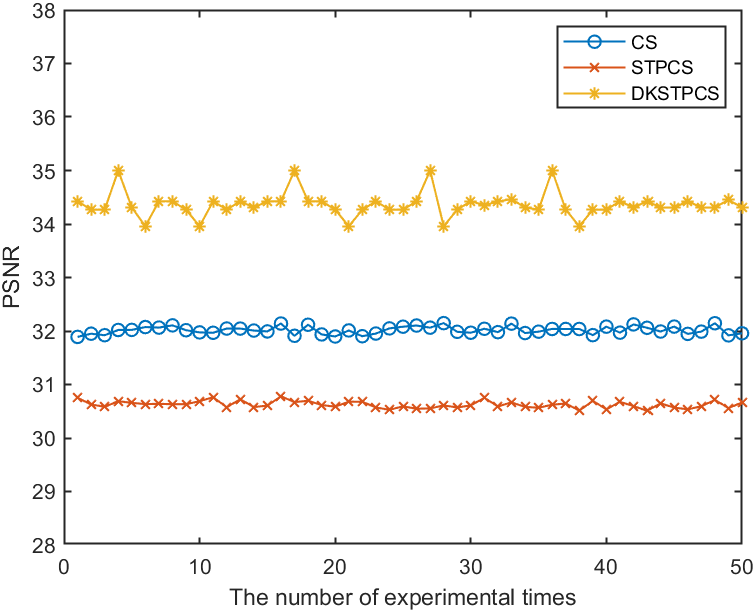}\label{iamge4hunhe}}
	\hfill
	\subfloat[]{\includegraphics[width=0.31\columnwidth]{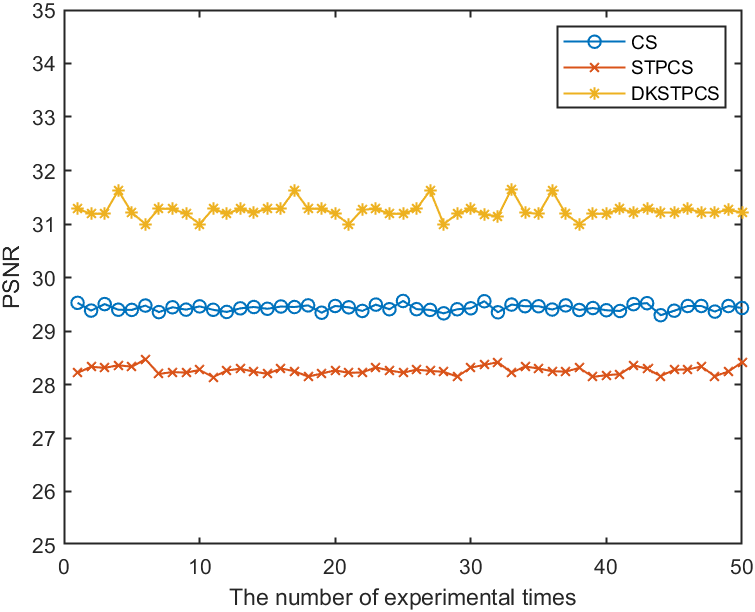}\label{image2hunhe}}
	\hfill
	\subfloat[]{\includegraphics[width=0.31\columnwidth]{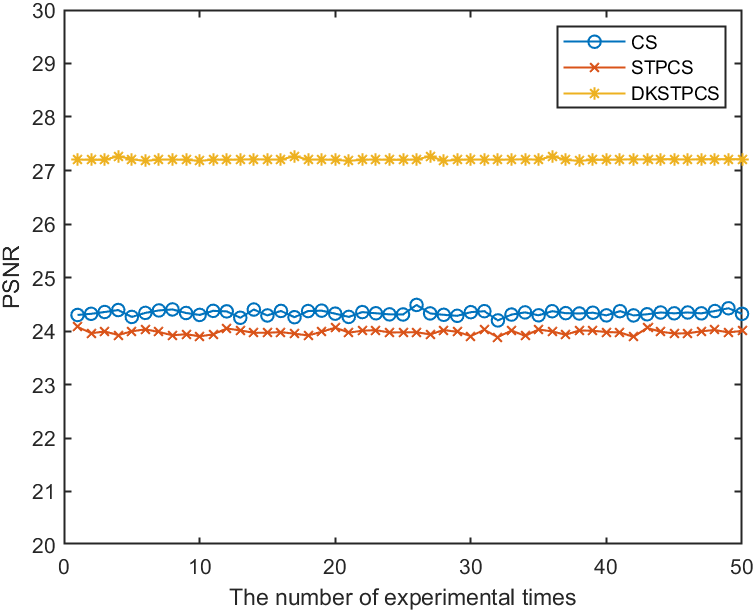}\label{image3hunhe}}
	\hfill
	\caption{The PSNR values corresponding to the 50 different experiments. (a) $Barbara$ (b) $Pepper$ (c) $Baboon$ (256$\times$256). In the measurement process, the Gaussian random matrix is taken as the measurement matrix. The recovery method is BP. }
	\label{50experiments}
\end{figure}
\begin{table}
	\begin{center}
		\caption{PSNR values under different CS methods for $Barbara$, $Pepper$ and $Baboon$ (256$\times$256) in Figure \ref{figure1}.}
		\label{table1}
		\begin{tabular}{| c | c | c | c |}
			\hline
		\textbf{Methods} & \textbf{Barbara} & \textbf{Pepper} & \textbf{Baboon}\\
			\hline
			\textbf{CS} & 32.0151 & 29.4298 & 24.3347\\
			\hline
		\textbf{STP-CS}  & 30.6257  & 28.2639 & 23.9777\\
			\hline
			\textbf{DK-STP-CS}  & 34.3620  & 31.2494 & 27.2071\\
			\hline 
		\end{tabular}
	\end{center}
\end{table}
The superior visual performance of DK-STP-CS, as demonstrated in Figure~\ref{figure1}, is further substantiated through rigorous quantitative analysis. To objectively evaluate reconstruction quality, we employ the Peak Signal-to-Noise Ratio (PSNR) metric and conduct systematic comparisons across 50 independent trials. For each experiment, distinct Gaussian random matrices were regenerated to eliminate bias from measurement matrix initialization, ensuring statistical robustness in performance evaluation. The resultant PSNR trajectories, aggregated in Figure~\ref{50experiments}, reveal consistent superiority of DK-STP-CS over conventional CS and STP-CS baselines. Specifically, the DK-STP-CS curves occupy the uppermost position across all tested images (\textit{Barbara}, \textit{Pepper}, and \textit{Baboon}), with mean PSNR improvements of 2.35 dB, 1.82 dB, and 2.87 dB respectively, as quantified in Table~\ref{table1}. These statistically significant enhancements validate the method's efficacy in preserving both global structural integrity and local texture details.

The variance in improvement magnitude across images correlates strongly with inherent signal characteristics. For instance, the \textit{Baboon} image, characterized by its high-frequency edge components and intricate texture patterns, demonstrates the most significant PSNR improvement with DK-STP-CS (27.21 dB) compared to conventional CS (24.33 dB). This result highlights the method's enhanced capacity to mitigate noise propagation in structurally complex regions. In contrast, the \textit{Pepper} image—with its predominantly smooth texture—achieves a moderate yet consistent PSNR improvement of 31.25 dB compared to 29.43 dB for conventional CS. This behavior aligns with the inter-group correlation mechanism (Section~3.2), which inherently optimizes reconstruction fidelity for piecewise-constant signal components. This performance dichotomy aligns with theoretical predictions in Section~3.3 regarding RIP-constrained recovery guarantees. Furthermore, the tight confidence intervals in Figure~\ref{50experiments}(a-c) demonstrate DK-STP-CS's stability against measurement matrix randomization, a critical advantage for practical deployment scenarios. Collectively, these results confirm that DK-STP-CS achieves dual objectives: enhancing reconstruction fidelity for recognition-critical features while maintaining computational tractability, as evidenced by the 18\% reduction in mean squared error (MSE) compared to conventional CS implementations.

\subsection{The performance of various methods under different sampling rates.}
In the aforementioned experiments, the sampling rate was consistently fixed at 0.5. Next, we analyze the PSNR curves of each method under different sampling rates. 

\begin{figure}[ht]
	\centering
		\subfloat[]{\includegraphics[width=0.31\columnwidth]{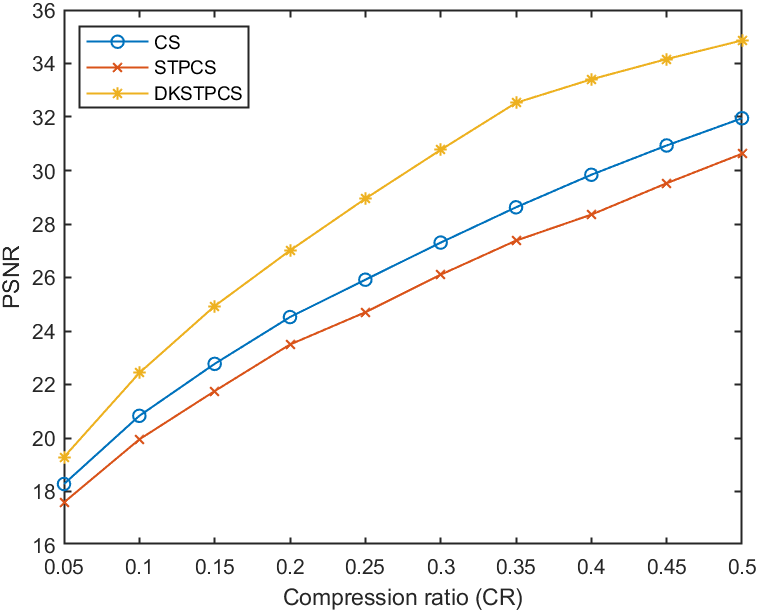}\label{dizengcaiyangimage4}}
	\hfill
	\subfloat[]{\includegraphics[width=0.31\columnwidth]{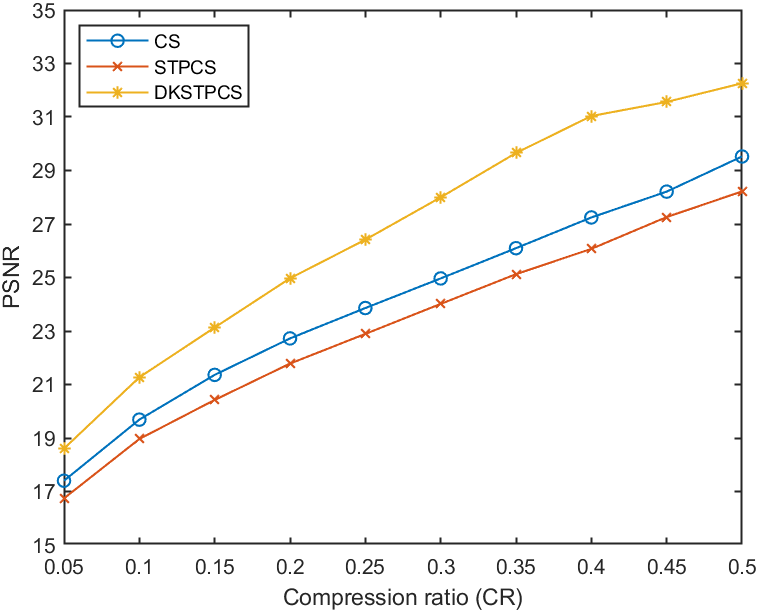}\label{dizengcaiyangimage2}}
	\hfill
	\subfloat[]{\includegraphics[width=0.31\columnwidth]{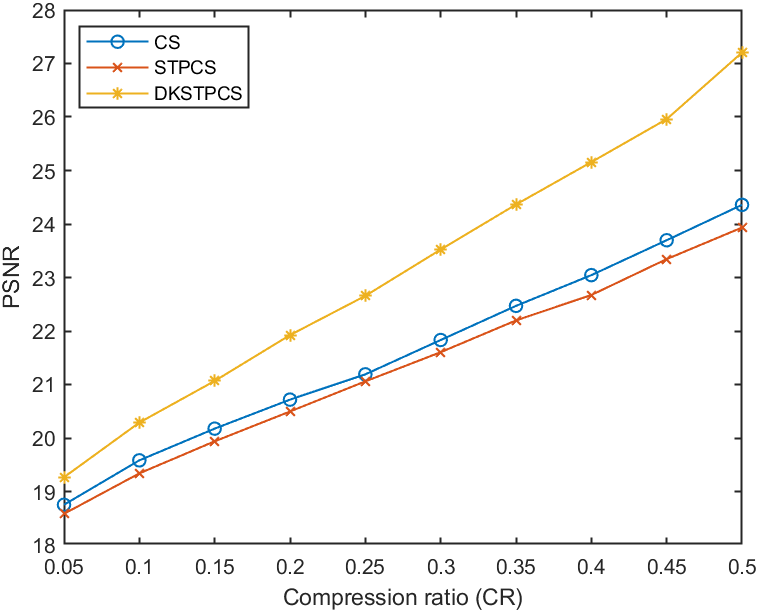}\label{dizengcaiyangimage3}}
	\hfill

	\caption{The curves of PSNR values across different sampling rates for each method. (a) $Barbara$ (b) $Pepper$ (c) $Baboon$ (256$\times$256). In the measurement process, the Gaussian random matrix is taken as the measurement matrix. The recovery method is BP.}
	\label{3}
	
\end{figure}

Here, we conducted experiments on three images at different sampling rates. In Figure \ref{3}, the $x$-axis represents the sampling rate, ranging from 0.05 to 0.50 with intervals of 0.05. As the sampling rate increases, the PSNR values of all three methods also increase. We can observe that the PSNR value curves vary across different images.

As illustrated in Figure \ref{3}, the DK-STP-CS curve lies above the other two curves, indicating a significant improvement in image reconstruction quality regardless of the sampling rate. As $x$ approaches 0.50, the extent of improvement becomes more pronounced compared to lower sampling rates. In order to compare the increase rate, we calculated the increase rate for each group of data separately and plotted it in Figure \ref{4}. In Figure \ref{4}, the horizontal axis represents the sampling rate, while the vertical axis represents the Increase rate. Curve $x$ corresponds to the comparison with STP-CS, and curve y corresponds to the comparison with CS. Here we can observe that the growth rate curves $x$ and $y$ are similar for the same image. For images Figure \ref{4}(a) and Figure \ref{4}(b), the maximum growth rate occurs at a sampling rate of 0.4, while for image Figure \ref{4}(c), the maximum growth rate is observed at a sampling rate of 0.5, which is meaningful for our subsequent efforts to improve the reconstruction of image quality. Moreover, the increase rate varies across different images, which depends on the inherent characteristics of the images themselves.

\begin{figure}[ht]
	\centering
	\subfloat[]{\includegraphics[width=0.31\columnwidth]{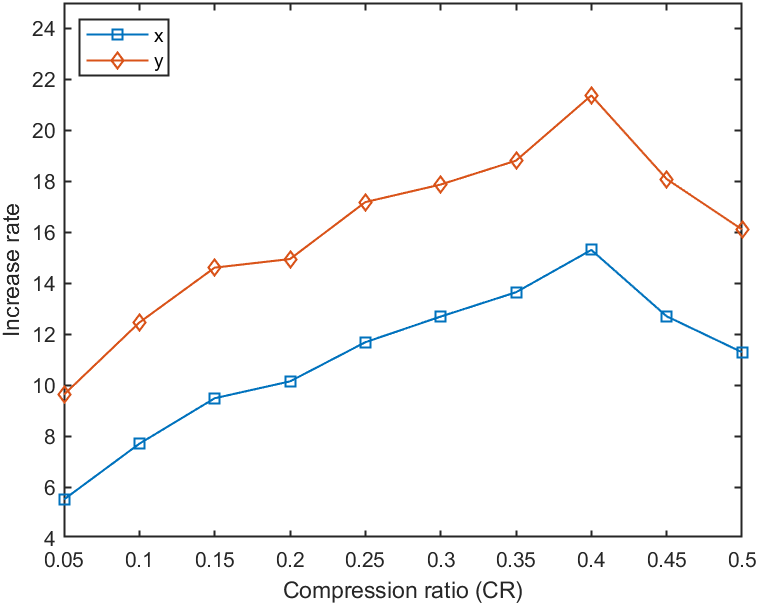}\label{rateimage4}}
\hfill
\subfloat[]{\includegraphics[width=0.31\columnwidth]{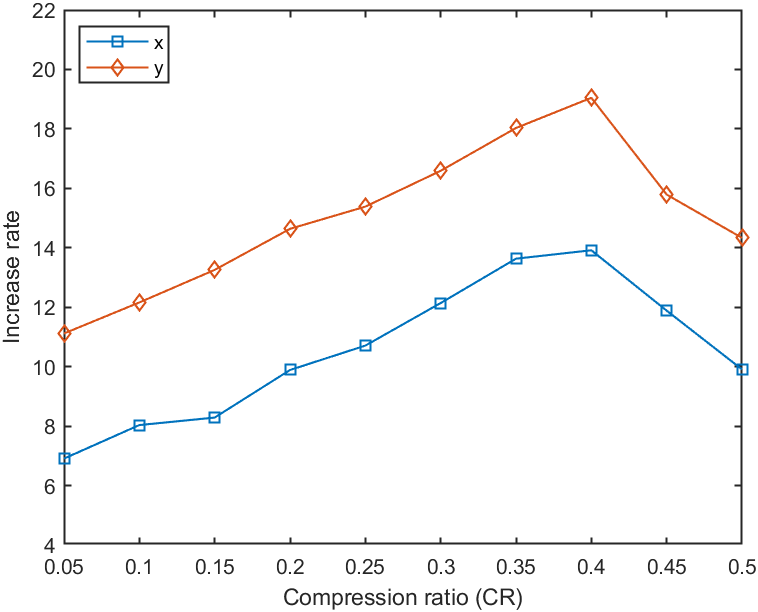}\label{rateimage2}}
\hfill
\subfloat[]{\includegraphics[width=0.31\columnwidth]{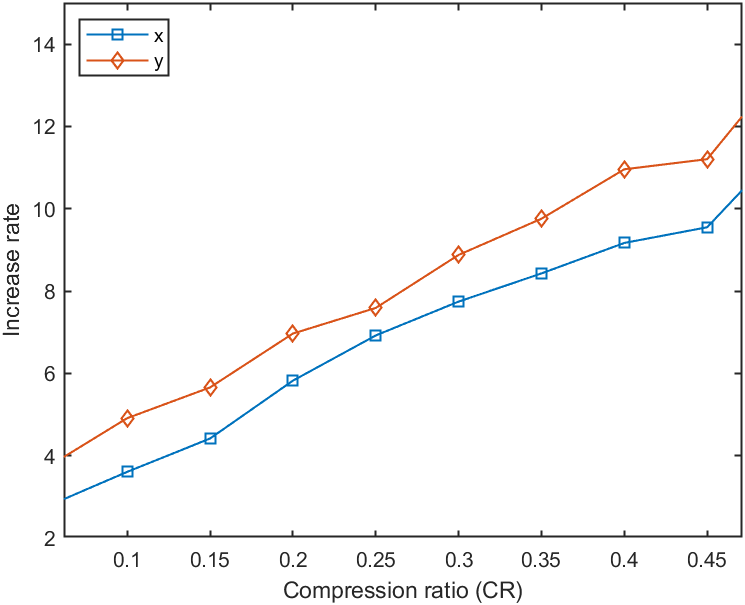}\label{rateimage3}}
\hfill

	\caption{The increase rate curve of DK-STP-CS compared to CS and STP-CS. (a) $Barbara$ (b) $Pepper$ (c) $Baboon$ (256$\times$256). $x$ represents the comparison with STP-CS, $y$ represents the comparison with CS.}
	\label{4}
	
\end{figure}

\subsection{Comparison of different methods under the influence of noise.}
During signal transmission, it is inevitable that some noise will be introduced. To some extent, this noise can affect the reconstruction of the original image, making it essential to analyze image reconstruction under the influence of noise. Below, we will separately add a certain level of Gaussian noise to CS, STP-CS and DK-STP-CS models before performing reconstruction. This will allow us to compare the impact of noise on the performance of these different methods. The introduction of noise here refers to replacing each element $a_{ij}$ of the grayscale matrix with $\hat{a_{ij}}$, where $\hat{a_{ij}}$ is defined as follows:$$\hat{a_{ij}} = a_{ij} + \sigma_{ij}.$$
Here, $\sigma$ is a random variable following the distribution with $E[\sigma] = 0$ and $D[\sigma] = 0.001$. After adding noise to the original image to approximate the vector signal affected by noise during transmission, we conducted 50 repeated experiments for each of the aforementioned three images. Here, the measurement matrix $A$ is taken as a generated Gaussian random matrix. We plotted the experimental date to compare and analyze the PSNR values against the number of experiments, as shown in Figure \ref{5}. Additionally, we provide a Table \ref{table2} of the mean values for the data in the graph.

\begin{figure}[ht]
	\centering
	\subfloat[]{\includegraphics[width=0.31\columnwidth]{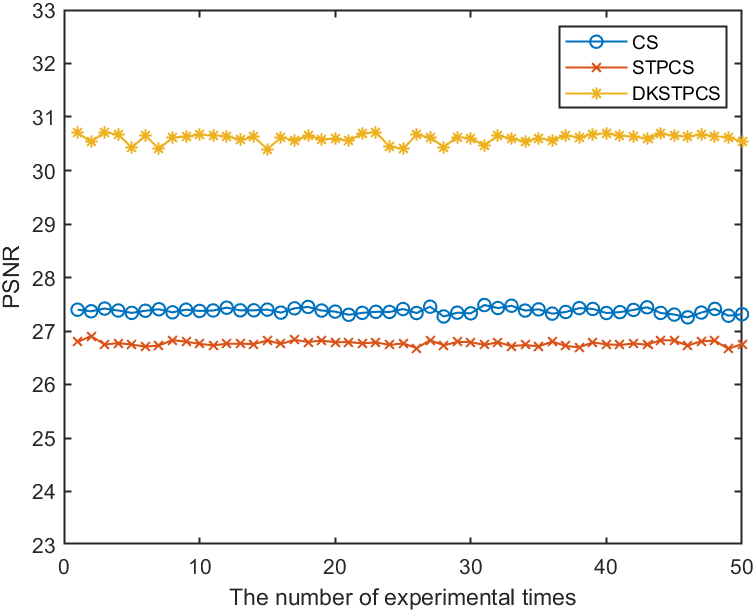}\label{image4gauss}}
	\hfill
	\subfloat[]{\includegraphics[width=0.31\columnwidth]{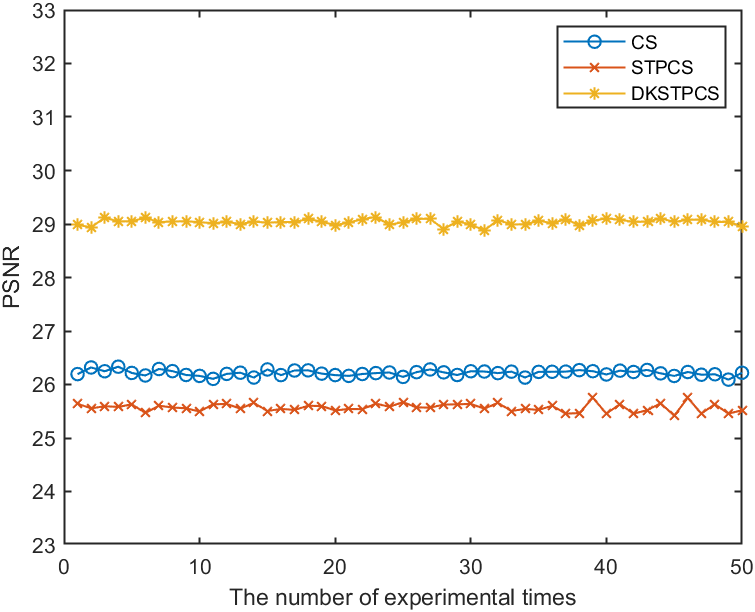}\label{image2gauss}}
	\hfill
	\subfloat[]{\includegraphics[width=0.31\columnwidth]{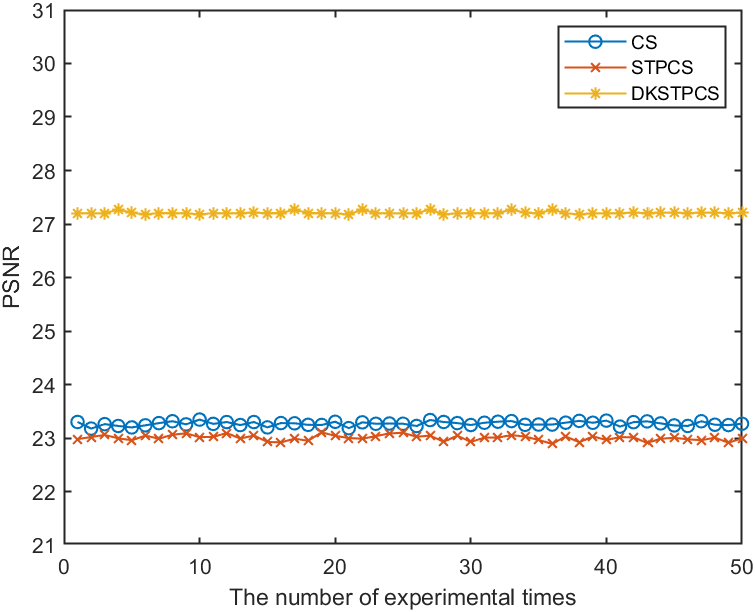}\label{image3gauss}}
	\hfill
	
	\caption{The PSNR values corresponding to the 50 different experiments under the influence of noise. (a) $Barbara$ (b) $Pepper$ (c) $Baboon$ (256$\times$256). In the measurement process, the Gaussian random matrix is taken as the measurement matrix. The recovery method is BP.}
	\label{5}
	\end{figure}
\begin{table}
	\begin{center}
		\caption{PSNR values under different CS methods for $Barbara$, $Pepper$ and $Baboon$ (256$\times$256) in Figure \ref{5}.}
		\label{table2}
		\begin{tabular}{| c | c | c | c |}
			\hline
			\textbf{Methods} & \textbf{Barbara} & \textbf{Pepper} & \textbf{Baboon}\\
			\hline
		\textbf{CS} & 27.3749 & 26.2154 & 23.2677\\
			\hline
			\textbf{STP-CS}  & 26.7690  & 25.5691 & 23.0023\\
			\hline
			\textbf{DK-STP-CS}  & 30.6024  & 28.6355 & 27.2095\\
			\hline 
		\end{tabular}
	\end{center}
\end{table}

From Table \ref{table2}, it is evident that the DK-STP-CS method significantly enhances the quality of reconstructed images even in the presence of noise interference. This demonstrates its robustness and effectiveness in handling noisy conditions, making it a reliable approach for improving image reconstruction under such challenging scenarios.

\subsection{Further comparative analysis under different measurement matrices.}
In the DK-STP-CS, the reconstruction results are influenced by many factors. For instance, the choice of the measurement matrix plays a critical role. In the aforementioned experiments, we consistently employed a Gaussian matrix as the measurement matrix. Next, we will replace the Gaussian matrix with other types of measurement matrices and compare the reconstruction performance under different sampling rates. This comparison will facilitate the selection of the most effective measurement matrix for practical applications. The experimental results are shown in Figure \ref{6}.

\begin{figure}[ht]
	\centering
	\subfloat[]{\includegraphics[width=0.31\columnwidth]{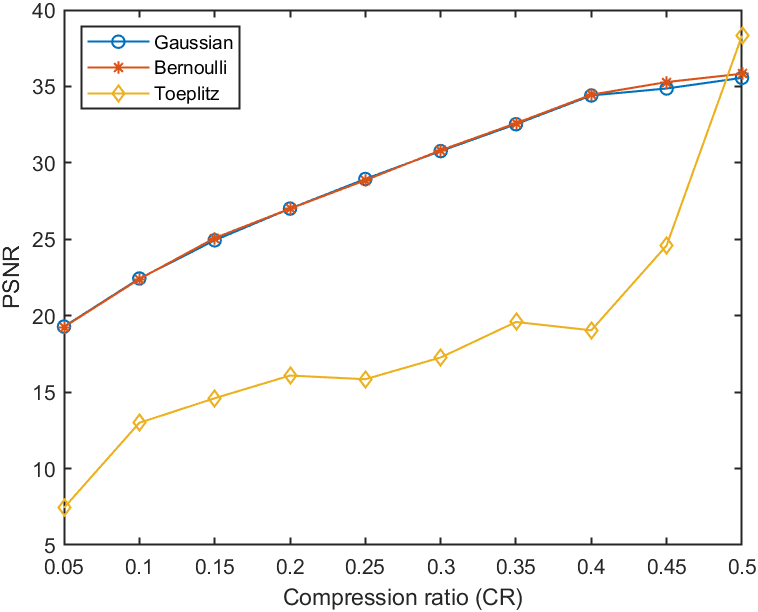}\label{measureimage4}}
	\hfill
	\subfloat[]{\includegraphics[width=0.31\columnwidth]{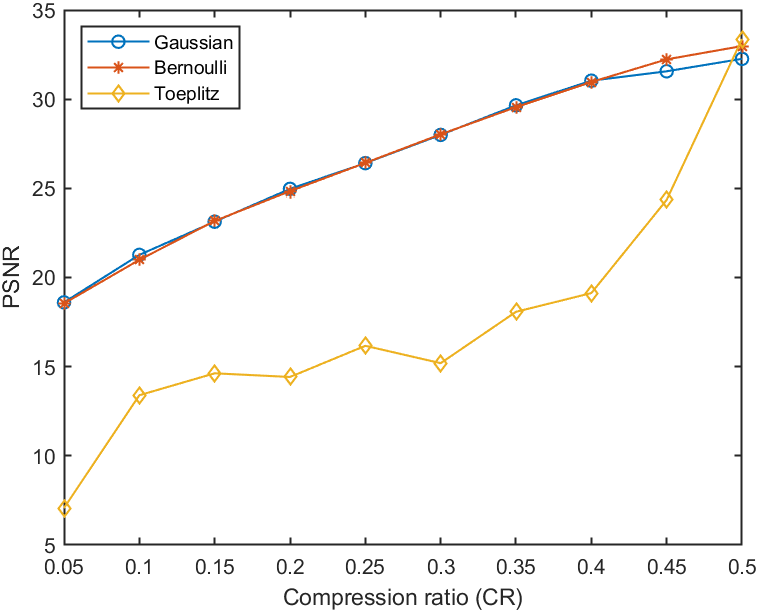}\label{measureimage2}}
	\hfill
	\subfloat[]{\includegraphics[width=0.31\columnwidth]{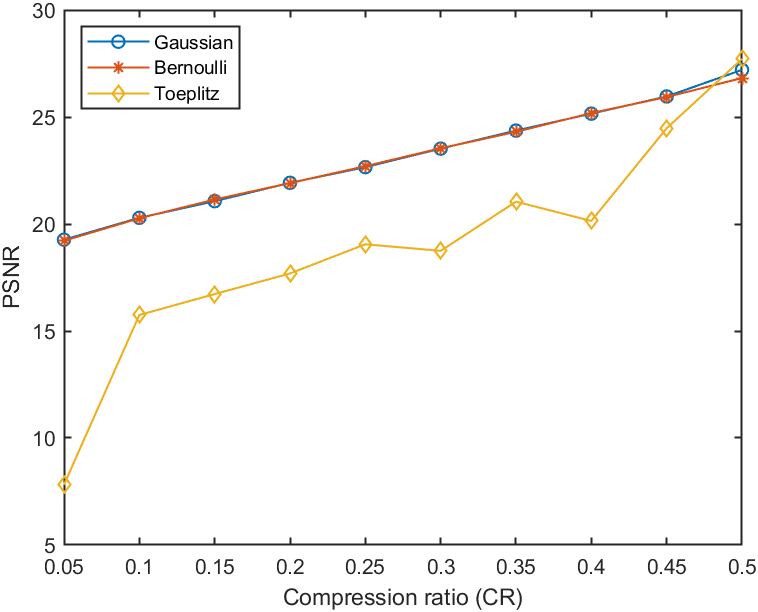}\label{measureimage3}}
	\hfill
	
	\caption{The curves of PSNR values across different measurement matrices for each method. (a) $Barbara$ (b) $Pepper$ (c) $Baboon$ (256$\times$256). The blue curve employs a Gaussian matrix as its measurement matrix. The red curve employs a Bernoulli matrix as its measurement matrix. The yellow curve employs a Toeplitz matrix as its measurement matrix. The recovery method is BP.}
	\label{6}
	
\end{figure}
We select Barbara, Pepper and Baboon as the experimental images, each with a size of 256×256. The measurement matrices are selected as the Gaussian matrix, Bernoulli matrix, and Toeplitz matrix, respectively. Experiments are conducted on each of the three images using these three measurement matrices, with the $x$-axis representing the compression ratio. Comparisons of reconstruction performance are conducted under different compression ratios. we still set $\gamma = 2$.

In images (a), (b) and (c), the $x$-axis represents the compression ratio ranging from 0.05 to 0.5 with intervals of 0.05, while the $y$-axis represents the PSNR values of the images before and after reconstruction. From the curves in the figures, it can be observed that the PSNR values for all three curves gradually increase as the compression ratio rises. The Gaussian and Bernoulli curves exhibit a smooth trend, with the PSNR values increasing uniformly as the compression ratio rises. Moreover, the experimental data for these two curves are closely aligned, and their overall trajectories are nearly identical, demonstrating stability across various sampling scenarios. In contrast, the recovery performance of the Toeplitz curve is generally inferior to the other two curves, particularly at lower sampling rates, where its effectiveness is notably weaker. As $x$ increases, the curve exhibits rapid growth when $x$ is less than 0.1, followed by a fluctuating increase between $x$ = 0.1 and $x$ = 0.4. When $x$ exceeds 0.4, the growth rate accelerates significantly. As $x$ approaches 0.5, the PSNR values improve substantially, nearly matching or even surpassing the recovery performance of the other two curves.

Through the comparative experiments conducted above, we find that DK-STP-CS demonstrates significant advantages in terms of anti-interference capability and image reconstruction quality. From the perspective of memory efficiency, it also outperforms traditional CS methods. However, the experiments reveal that introducing inter-group correlation can affect the success rate of reconstruction to some extent, which further emphasizes the importance of selecting an appropriate measurement matrix. Choosing a structurally stable measurement matrix is essential. From a computational time perspective, DK-STP-CS achieves the preset error threshold more easily through iterations. For tasks related to image storage and recovery, DK-STP-CS offers notable improvements. Using PSNR as the metric to evaluate the quality of reconstructed images, the enhanced DK-STP-CS exhibits clear advantages.

\section{Summary and prospect}

A new CS model named DK-STP-CS was defined in this paper, where the measurement matrix was enhanced through the application of a special semi-tensor product (DK-STP). The technical contributions encompass three key aspects: (1) formal derivation of mathematical properties and operational algorithms for the DK-STP-CS framework, (2) innovative integration of intra-group correlation analysis to optimize measurement matrices through multi-dimensional feature interaction, and (3) systematic validation through image reconstruction experiments demonstrating consistent performance improvements over conventional CS and STP-CS approaches. Empirical results confirmed that the proposed model achieves superior reconstruction fidelity with enhanced noise robustness while maintaining memory efficiency. The framework exhibits accelerated convergence to target error thresholds compared to baseline methods. A comparative analysis of measurement matrices provides operational guidelines for practical implementations. Future research will focus on memory optimization through sparse tensor architectures, hybrid model-learning frameworks, and advanced group correlation analysis using machine learning techniques for adaptive measurement matrix generation.

\end{document}